%% file: main.tex
\documentclass[10pt]{llncs}

\titlerunning{A Temporal Logic for Asynchronous Hyperproperties}

\newcommand{\Thanks}{\thanks{This work was funded in part by Madrid
    Regional Government under project ``S2018/TCS-4339 (BLOQUES-CM)'',
    and by Spanish National Project ``BOSCO
    (PGC2018-102210-B-100)''.}}

\title{A Temporal Logic for\\ Asynchronous Hyperproperties\Thanks}

\author{Jan Baumeister\inst{1} \and Norine Coenen\inst{1} \and\\ Borzoo Bonakdarpour\inst{2} \and Bernd Finkbeiner\inst{1} \and
  C\'esar S\'anchez\inst{3}\orcidID{0000-0003-3927-4773}}
\institute{CISPA Helmholtz Center for Information Security, Germany \\
  \and Michigan State Univerrsity, USA \\ 
  \and IMDEA Software Institute, Spain
}

\input{preamble}

\newcommand{\SPICorrect}{\textsf{SPI-correct}\xspace}
\newcommand{\SPITerm}{\textsf{SPI-term}\xspace}

\begin{document}

\maketitle

\input{abstract}

\input{intro}
\input{prelim}
\input{ahltl}
\input{algorithm}
\input{colors}
\input{undec}
\input{casestudies}
\input{related}
\input{conclusion}

\vfill\pagebreak
%


\bibliographystyle{abbrv}
\bibliography{bibliography}

\vfill
\pagebreak
\appendix

\input{proofs}
\input{decproofs}

\end{document}

%% file: preamble.tex
\usepackage{amsmath}
\usepackage{graphicx}
\usepackage{amsfonts, amssymb}
\usepackage{xspace}
\usepackage{algorithm}
\usepackage{algpseudocode}
\usepackage{ltlfonts}
\usepackage{latexsym, stmaryrd, xspace}
\usepackage{paralist}
\usepackage{todonotes}
\usepackage{array}
\usepackage{hyperref}
\hypersetup{%
	colorlinks=true,           
	allcolors=blue!70!black,   
	pdfstartview=Fit,          
	breaklinks=true,
	pdfauthor={J. Baumeister, N. Coenen, B. Bonakdarpour, B. Finkbeiner, and C. Sanchez},
	pdftitle={A Temporal Logic for Asynchronous Hyperproperties}}

\usepackage{wrapfig}
\usepackage{listings}
\usepackage{cite}%
\usepackage{array}%
\usepackage{mdwtab}

\usepackage{wrapfig}
\usepackage{caption}
\usepackage{subcaption}
\usepackage{siunitx}

\input{macros/macros}

\input{async-macros}

\definecolor{mGreen}{rgb}{0,0.6,0}
\definecolor{mGray}{rgb}{0.5,0.5,0.5}
\definecolor{mPurple}{rgb}{0.58,0,0.82}
\definecolor{backgroundColour}{rgb}{0.95,0.95,0.92}

\lstdefinestyle{CStyle}{
    backgroundcolor=\color{backgroundColour},   
    commentstyle=\color{mGreen},
    keywordstyle=\color{magenta},
    numberstyle=\tiny\color{mGray},
    stringstyle=\color{mPurple},
    basicstyle=\footnotesize,
    breakatwhitespace=false,         
    breaklines=true,                 
    captionpos=b,                    
    keepspaces=true,                 
    numbers=left,                    
    numbersep=2pt,                  
    showspaces=false,                
    showstringspaces=false,
    showtabs=false,                  
    tabsize=2,
    language=C
}

\makeatletter
    \def\@citecolor{blue}%
    \def\@urlcolor{blue}%
    \def\@linkcolor{blue}%
 
 \def\orcidID#1{\smash{\href{http://orcid.org/#1}{\protect\raisebox{-1.25pt}{\protect\includegraphics{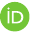}}}}}
 \makeatother

%% file: macros/macros.tex
\input{macros/hyper}

\input{macros/ltl}

\input{macros/logic}

\newcommand{\Dom}{\ensuremath{\textit{Dom}}}

\newcommand{\ie}{i.e.,\xspace}

%% file: macros/hyper.tex
\newcommand{\Vars}{\ensuremath{\mathsf{Vars}}\xspace}
\newcommand{\Traces}{\ensuremath{\mathsf{Traces}}\xspace}

\newcommand{\AP}{\ensuremath{\mathsf{AP}}\xspace}

\newcommand{\PTR}{\ensuremath{\mathsf{PTR}}\xspace}

\newcommand{\V}{\mathcal{V}}

\newcommand{\state}{s}
\newcommand{\trace}{\sigma}
\newcommand{\traj}{t}
\newcommand{\trvar}{\pi}
\newcommand{\tjvar}{\tau}
\newcommand{\trass}{\Pi}

\newcommand{\tjall}{\ensuremath{\mathsf{TRJ}}\xspace}
\newcommand{\alphabet}{\mathrm{\Sigma}}
\newcommand{\naturals}{\mathbb{N}_0}

\newcommand{\term}{\mathsf{term}}

\newcommand{\vphi}[1]{\ensuremath{\varphi_{\textit{#1}}}}

\newcommand{\varphiNoPCP}{\varphi_{\overline{\mathit{pcp}}}}
\newcommand{\varphiType}{\vphi{type}}
\newcommand{\varphiSeqWV}{\varphi_{\mathit{domino}}}
\newcommand{\varphiLetter}{\varphi_{\mathit{word}}}

\newcommand{\kripPCP}{\krip_{\mathit{pcp}}}
\newcommand{\psiMap}{\psi_{\mathit{map}}}
\newcommand{\psiSync}{\psi_{\mathit{sync}}}

\newcommand{\krip}{\mathcal{K}}
\newcommand{\ktuple}{\langle S, S_\init, \trans, L \rangle}
\newcommand{\init}{\mathit{init}}

\newcommand{\Trace}{\mathsf{Traces}}
\newcommand{\States}{S}
\newcommand{\trans}{\delta}

\newcommand{\idx}{\ensuremath{\mathit{dom}}\xspace}
\newcommand{\lc}{\ensuremath{\mathit{lc}}\xspace}

\newcommand{\comp}[1]{\textsf{\small #1}}

\newcommand{\sync}{\mathit{sync}}

\newcommand{\Atau}{\mathsf{A}}
\newcommand{\Etau}{\mathsf{E}}

\newcommand{\Domi}[2]{\Big[\frac{#1}{#2}\Big]}
\newcommand{\domi}[2]{[\frac{#1}{#2}]}

\newcommand{\psiPH}{\psi_{\textit{ph}}}
\newcommand{\kripST}{\krip^{\textit{st}}}
\newcommand{\ST}{\ensuremath{\mathit{st}}}
\newcommand{\SST}{S^\ST}

\newcommand{\transST}{\trans^\ST}
\newcommand{\LST}{L^\ST}

\newcommand{\Acc}{\ensuremath{\mathit{acc}}}
\newcommand{\sbot}{\ensuremath{s_\bot\xspace}}
\newcommand{\SAcc}{S^\Acc}

\newcommand{\transAcc}{\trans^\Acc}
\newcommand{\LAcc}{L^\Acc}

\newcommand{\psiESync}{\psi_{\textit{esync}}}

\newcommand{\PS}{\mathcal{P}}
\newcommand{\Change}{\ensuremath{\textit{change}\xspace}}
\newcommand{\Move}{\ensuremath{\textit{move}\xspace}}
\newcommand{\Phases}{\ensuremath{\textit{phase}\xspace}}

\newcommand{\Fair}{\ensuremath{\textit{fair}\xspace}}

\newcommand{\Bad}{\ensuremath{\textit{missalign}\xspace}}
\newcommand{\Align}{\ensuremath{\textit{align}\xspace}}
\newcommand{\Block}{\ensuremath{\textit{block}\xspace}}
\newcommand{\Cycles}{\ensuremath{\textit{cycles}\xspace}}
\newcommand{\Expand}{\ensuremath{\textit{expand}\xspace}}
\newcommand{\Compress}{\ensuremath{\textit{compress}\xspace}}

\newcommand{\acc}{\ensuremath{\textit{acc}}}
\newcommand{\dec}{\ensuremath{\textit{dec}}}
\newcommand{\kripAcc}{\krip^{\textit{acc}}}

\newcommand{\Color}{\ensuremath{\textit{color}}}

\newcommand{\lnz}{\varphi_{\mathsf{LNZ}}}
\newcommand{\history}{\mathsf{history}}
\newcommand{\gmni}{\varphi_{\mathsf{GMNI}}}
\newcommand{\nnt}{\varphi_{\mathsf{NNT}}}
\newcommand{\tin}{\varphi_{\mathsf{TIN}}}
\newcommand{\tsn}{\varphi_{\mathsf{TSN}}}
\newcommand{\subst}{\triangleleft}

\newcommand{\Prog}{\textit{P}}
\newcommand{\ProgOne}{\ensuremath{\Prog{}_1}\xspace}
\newcommand{\ProgTwo}{\ensuremath{\Prog{}_2}\xspace}

%% file: macros/ltl.tex
\newcommand{\U}{\mathbin\mathcal{U}}
\newcommand{\G}{\LTLsquare}

\newcommand{\F}{\LTLdiamond}

\newcommand{\Always}{\LTLsquare}
\newcommand{\Event}{\LTLdiamond} 
\newcommand{\Next}{\LTLcircle}

\newcommand{\PrevNoFirst}{\LTLcircletilde}
\newcommand{\HasAlwaysBeen}{\LTLsquareminus}

\newcommand{\Since}{\mathbin{\mathcal{S}}}

%% file: macros/logic.tex
\newcommand{\true}{\mathit{true}}
\newcommand{\false}{\mathit{false}}
\newcommand{\into}{\ensuremath{\rightarrow}}
\renewcommand{\And}{\mathrel{\wedge}}
\newcommand{\Or}{\mathrel{\vee}}

\newcommand{\Then}{\mathrel{\rightarrow}} 
\newcommand{\Iff}{\mathrel{\leftrightarrow}}

\newcommand{\Into}{\mathrel{\rightarrow}}
\newcommand{\IntoP}{\mathrel{\rightharpoonup}}
\newcommand{\partTo}{\IntoP}
\newcommand{\Nat}{\mathbb{N}}
 
\newcommand{\tupleof}[1]{\langle#1\rangle}

%% file: async-macros.tex
\newcommand{\AHLTLName}{\textup{A-HLTL}}
\newcommand{\AsyncHyperLTL}{\AHLTLName\xspace}
\newcommand{\AHLTL}{\AsyncHyperLTL}

\newcommand{\defAs}{\ensuremath{\stackrel{\text{\textup{def}}}{=}}}
\newcommand{\DefinedAs}{\ensuremath{\,\defAs\,}}

\newcommand{\DefOR}{\ensuremath{\hspace{0.6em}\big|\hspace{0.6em}}}

\newcommand{\li}{\mathit{li}}
\newcommand{\lo}{\mathit{lo}}

\newcommand{\phiOD}{\varphi_{\mathsf{OD}}\xspace}

%% file: abstract.tex
\begin{abstract}
  {\em Hyperproperties} are properties of computational systems that
  require more than one trace to evaluate, e.g., many information-flow
  security and concurrency requirements.
  Where a trace property defines a set of traces, a hyperproperty
  defines a set of sets of traces.
  The temporal logics HyperLTL and HyperCTL* have been proposed to
  express hyperproperties.
  However, their semantics are {\em synchronous} in the sense that all
  traces proceed at the same speed and are evaluated at the same
  position.
  This precludes the use of these logics to analyze systems whose
  traces can proceed at different speeds and allow that different
  traces take stuttering steps independently.
  To solve this problem in this paper, we propose an {\em
    asynchronous} variant of HyperLTL.
  On the negative side, we show that the model-checking problem for
  this variant is undecidable.
  On the positive side, we identify a decidable fragment which covers
  a rich set of formulas with practical applications.
  We also propose two model-checking algorithms that reduce our
  problem to the HyperLTL model-checking problem in the synchronous
  semantics.
\end{abstract}


%% file: intro.tex
\section{Introduction}
\label{sec:intro}

Hyperproperties~\cite{cs10} extend the conventional notion of trace
properties~\cite{as85} from a set of traces to a set of sets of
traces.
In other words, a hyperproperty stipulates a system property and not
the property of just individual traces.
Many interesting requirements in computing systems are hyperproperties
and cannot be expressed by trace properties.
Examples include
\begin{inparaenum}[(1)]
\item a wide range of information-flow security policies such as {\em
    noninterference}~\cite{gm82} and {\em observational
    determinism}~\cite{zm03},
\item sensitivity and robustness requirements in cyber-physical
  systems~\cite{wzbp19}, and
\item consistency conditions such as {\em linearizability} in
  concurrent data structures~\cite{bss18}.
\end{inparaenum}

HyperLTL~\cite{clarkson14temporal} is a temporal logic for
hyperproperties that enriches LTL with quantifiers allowing explicit
and simultaneous quantification over multiple execution traces.
For example, the observational determinism security
policy~\cite{zm03} stipulates that any two executions that
start in two {\em low-equivalent} states (i.e., states whose value of
publicly observable variables are the same), should remain in
low-equivalent states.
%
%
%
This property can be expressed in HyperLTL as the following formula,
called $\phiOD$,
\( \forall \pi. \forall \pi'. (l_\pi \leftrightarrow l_{\pi'}) \,
\rightarrow \, \G(l_\pi \leftrightarrow l_{\pi'}).  \)
However, the semantics of HyperLTL (and other formal languages for
hyperproperties) is {\em synchronous}, meaning that they completely
abstract away the notion of time passage.
In HyperLTL, all traces proceed at the same speed, as all temporal
operators move the position on all traces simultaneously.
Consider the program \ProgOne in Fig~\ref{fig:prog1}, where input
values $\mathtt{0}$ and $\mathtt{1}$ are possible for {\em
  high-secret} variable $\mathtt{h}$.
This renders two possible traces shown in~Fig.~\ref{fig:sync1} that
satisfy $\phiOD$.

\renewcommand*\thelstnumber{\ensuremath{\ell_{\the\value{lstnumber}}}}
\begin{figure}[t]
	\begin{minipage}{.3\textwidth}
          \begin{lstlisting}[style=CStyle,tabsize=2,language=ML, basicstyle=\scriptsize\ttfamily,
            % ,escapechar=/
            ]
 int l = 0;
			
 if (h = 0)
     l := l + 1;
 else
     l := 1;

\end{lstlisting}
     \vspace{0.5em}
		\caption{Program \ProgOne}
		\label{fig:prog1}
	\end{minipage}
\hfill
\begin{minipage}{.3\textwidth}
          \begin{lstlisting}[style=CStyle,tabsize=2,language=ML, basicstyle=\scriptsize\ttfamily,
			%    ,escapechar=/
			]
 int l = 0;
	
 if (h = 0)
     reg:= l + 1;
     l := reg;
 else
     l := 1;\end{lstlisting}
    \vspace{-0.4em}
    \caption{Program \ProgTwo}
    \label{fig:prog2}
  \end{minipage}
\hfill
\begin{minipage}{.3\textwidth}
  \vspace{0.2em}
  \includegraphics[scale=.45]{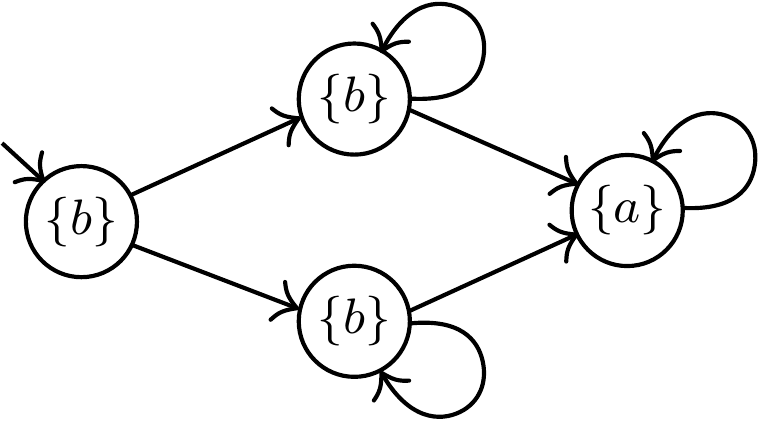}
  \vspace{0.2em}
       \caption{$\krip$ with a self-loop}
\label{fig:selfloop}
\end{minipage}
\vspace{-0.5em}	
\end{figure}

%

The synchronous semantics of HyperLTL has a shortcoming which has
practical implications as well: formulas are not invariant under {\em
  stuttering}.
Note that, contrary to LTL, disallowing the use of $\Next$ does not
make the formula invariant under stuttering different, as traces can
still stutter independently.
This limits the scope of application of HyperLTL to only those
settings where different traces can be perfectly aligned.
For example, consider program \ProgTwo in Fig.~\ref{fig:prog2}, where
line $\ell_4$ in \ProgOne is refined to its intermediate code using a
register that stores the value $\mathtt{l+1}$ and then stores this
value in memory location $\mathtt{l}$ in lines $\ell_4$ and $\ell_5$,
respectively.
Applying the synchronous semantics of HyperLTL results in declaring a
violation of $\phiOD$ in the second position.
This, however, is not an accurate interpretation of $\phiOD$ (assuming
that an attacker only has access to the memory footprint and not the
CPU registers or a timing channel), as the two traces are stutter
equivalent with respect to the state of variable $\mathtt{l}$.
In fact, the synchronous semantics of HyperLTL may incorrectly
identify good programs as bad because it ignores the notion of
relative time between traces.
This problem is generally amplified in Kripke structures where
self-loops correspond to non-deterministic choices that model that the
system may remain in a state for some arbitrary time.
For instance, consider $\krip$ in Fig.~\ref{fig:selfloop} and HyperLTL
formula
$\forall \pi.\forall \pi'. ((b_\pi \leftrightarrow b_{\pi'}) \, \U \,
\G(a_\pi \leftrightarrow a_{\pi'}))$.
Only pairs of traces that take the self-loop the same number of times
satisfy this formula.
However, since the goal of employing a self-loop is typically to make the 
duration of staying in a state irrelevant, this semantics is too restrictive.

Besides HyperLTL, other logics have been proposed that allow trace
quantification, for example,
$H_\mu$~\cite{DBLP:journals/pacmpl/GutsfeldMO21}, which adds trace
quantifiers to the polyadic
$\mu$-calculus~\cite{Andersen/1994/APolyadicModalMuCalculus}.
For $H_\mu$, the model-checking problem is in general undecidable, but two fragments, the 
$k$-synchronous, $k$-context bounded fragments, have been identified for which model
checking remains decidable~\cite{DBLP:journals/pacmpl/GutsfeldMO21}.
It is not known, however, if any of the commonly used hyperproperties,
like observational determinism, noninterference, or linearizability,
can be encoded in these fragments.

In this paper, we propose an asynchronous temporal logic for
hyperproperties.
Our main motivation is to be able to reason about execution traces
according to the relative order of the sequences of actions in each
trace but not about the duration of each action.
Software is inherently asynchronous, and so is hardware in many cases
if one abstracts the execution platform or many features of the
execution platform like pipelines, caches, memory contention, etc.
We call our temporal logic {\em Asynchronous HyperLTL} or in short,
\AHLTL.
The key addition is the notion of~\emph{trajectory} that controls the
relative speed at which traces progress by chosing at each instant
which traces move and which traces stutter.
%
\begin{figure}[t!]
  \begin{subfigure}[b]{0.27\textwidth}
    \centering
    \includegraphics[scale=.44]{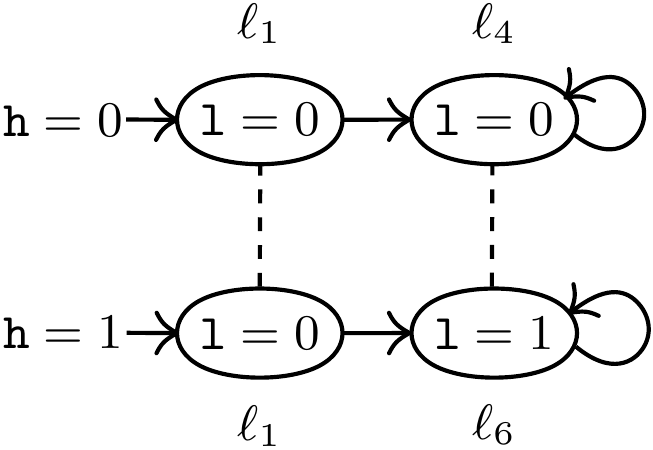}
    \caption{\ProgOne satisfies $\varphi_{\mathsf{OD}}$ under synchronous sems.}
    \label{fig:sync1}
  \end{subfigure}
  \hspace{0.7em}
  \begin{subfigure}[b]{0.31\textwidth}
    \centering
    \includegraphics[scale=.44]{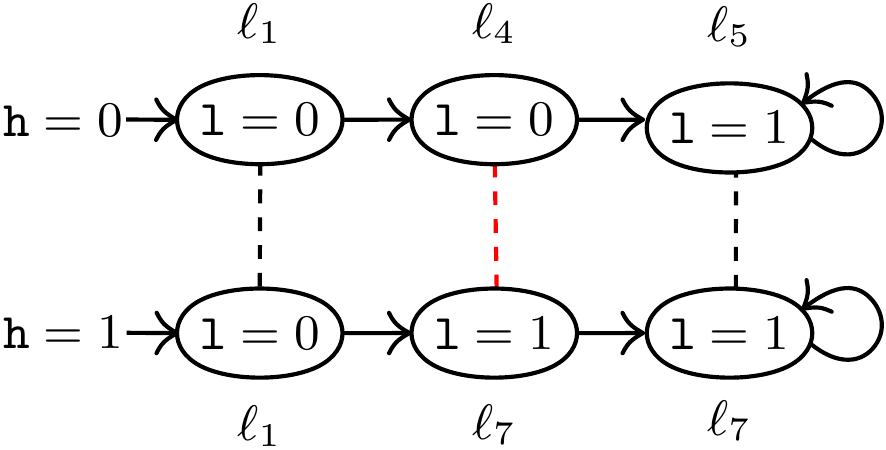}
    \caption{\ProgTwo violates  $\varphi_{\mathsf{OD}}$ under synchronous semantics}
    \label{fig:sync2}
  \end{subfigure}
  \hfill
  \begin{subfigure}[b]{0.31\textwidth}
    \centering
    \includegraphics[scale=.44]{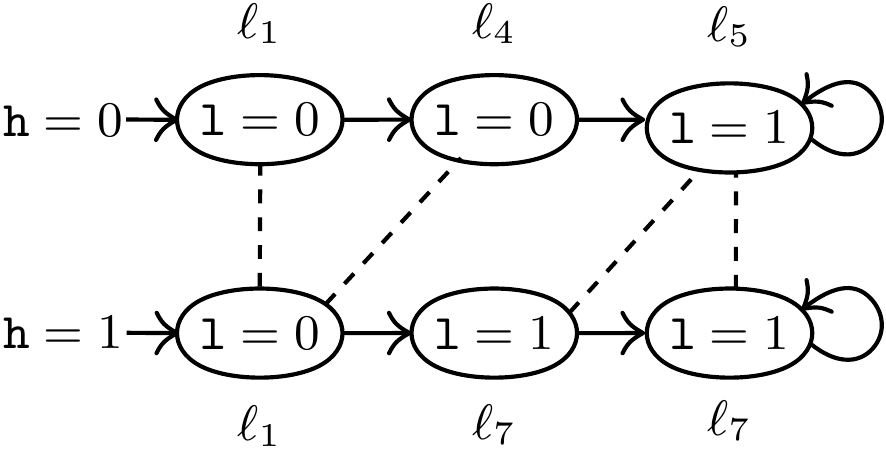}    
    \caption{\ProgTwo satisfies $\varphi_{\mathsf{OD}}$ under asynchronous semantics.}
    \label{fig:async}
  \end{subfigure}
  \caption{Synchronous vs. asynchronous semantics for HyperLTL.}
  \label{fig:async-sync}
\end{figure}
For example, the trajectory shown in Fig.~\ref{fig:async} for the two
traces of the program in Fig.~\ref{fig:prog2} allows the lower trace
to stutter in the first position while the upper trace advances.
On the contrary, in the third position, the upper trace stutters while the lower trace moves from the 
second to the third position.
This trajectory enables identification of stutter equivalence of the
two traces with respect to state variable $\mathtt{l}$ and, hence,
successful verification of observational determinism.
In order to reflect the notion of trajectories in our logic, we lift
the syntax of HyperLTL by allowing a trajectory modality.
This way, the corresponding formula for observational determinism in
\AHLTL is the following:
\[
\phiOD \DefinedAs \forall \pi.\forall \pi'.\Etau. (\li_\pi \Iff \li_{\pi'}) \Into
\G (\lo_\pi \Iff \lo_{\pi'})
\]
\label{form:obsdet}
where $\Etau$ denotes the {\em existence} of a trajectory for temporal
operator $\G$.
%
%
The \AHLTL formula for the Kripke structure in Fig.~\ref{fig:selfloop} is $\forall \pi.\forall \pi'. \Etau. 
((b_\pi \leftrightarrow b_{\pi'}) \, \U \, \G(a_\pi \leftrightarrow a_{\pi'}))$. 
\AHLTL allows us to reason about relational properties between
two different systems that differ on timing, like for example,
translation validation~\cite{pss98}, which relates executions of the
target code with the source code with respect to a (trace or hyper)
property.

We show an encoding of the PCP problem into model-checking a formula
of the shape
$ \forall \pi. \forall \pi'.\Etau. (\G\psi_1(\pi, \pi') \wedge
\F\psi_2(\pi, \pi'))
$, which implies that model-checking \AHLTL is undecidable, even for
the universal fragment.
On the positive side, we show two decidable fragments of \AHLTL.
The first algorithm is based on a {\em stuttering construction} in
which we modify the Kripke structure to accept all stuttering
expansions of the original paths.
This algorithm can handle fragment
$\forall \pi_1 \ldots \pi_n. \Etau.\psi$, where the $\psi$ is a
\emph{phase formula}, a class of safety formulas that appear in many
hyperproperties and are the building block of expressing trace
equivalence.
Our second algorithm uses an {\em acceleration construction} to
convert a finite sequence of transitions that do not change phase,
into a single transition.
This algorithm is able to handle formulas with arbitrary
quantification but a simpler kind of phase formulas.
\AHLTL is, thus, the first logic for hyperproperties that can express the major asynchronous 
hyperproperties of interest within decidable fragments.
Moreover, \AHLTL is the first logic for asynchronous hyperproperties
with a practical model checking algorithm.
Both algorithms use internally HyperLTL model-checking as a building
block.
However, the reduction from A-HLTL model-checking into HyperLTL
requires modifying both the formula and the model in a highly
non-trivial way, to encode the exitence of trajectories.
The choice of using HyperLTL model-checking as a building block is
based on the existence of tools, but it does not imply that
asynchronous properties of interest can be expressed in HyperLTL
directly.

We have evaluated the stuttering construction on two sets of cases
studies: a range of compiler optimizations and an SPI bus protocol.
In both case studies, we were able to prove system correctness using
our reduction from \AHLTL to synchronous HyperLTL.

\paragraph{Organization.} The rest of the paper is structured as follows.
Section~\ref{sec:prelim} contains the preliminaries, and
Section~\ref{sec:ahltl} introduces \AHLTL and presents examples of
properties expressible in \AHLTL.
Sections~\ref{sec:decidable} describes the decidable fragments and
present procedures for the model-checking problem.
Section~\ref{sec:undec} shows that the model-checking problem for
general \AHLTL formulas is undecidable and present the lower-bound
complexity.
Experimental results are presented in Section~\ref{sec:eval}.
Finally, Section~\ref{sec:related} discusses the related work, while
Section~\ref{sec:conclusion} concludes. Detailed proofs appear in the
appendix.


%

%% file: prelim.tex
\section{Preliminaries}
\label{sec:prelim}

Let $\AP$ be a set of {\em atomic propositions} and $\alphabet=2^\AP$
be the {\em alphabet}, where we call each element of $\alphabet$ a
{\em letter}.
A {\em trace} is an infinite sequence $\sigma=a_0a_1\cdots$ of letters from 
$\alphabet$.
We denote the set of all infinite traces by $\alphabet^\omega$.
We use $\sigma(i)$ for $a_i$ and $\sigma^i$ for the suffix
$a_ia_{i+1}\cdots$.
A {\em pointed trace} is a pair $(\sigma,p)$, where $p \in \naturals$ is a
natural number (called the {\em pointer}).
Pointed traces allow to traverse a trace by moving the pointer.
Given a pointed trace $(\sigma, p)$ and $n > 0$, we use
$(\sigma, p) + n$ as a short for $(\sigma,p+n)$.
We denote the set of all pointed traces by
$\PTR = \{(\trace, p) \mid \trace \in \alphabet^\omega \, \text{ and } \, p
\in \naturals\}$.

Two pointed traces $(\sigma,p)$ and $(\sigma',p')$ are \emph{stuttering
equivalent} if there are two infinite sequences of indices
$p=i_0<i_1\ldots$ and $p'=j_0<j_1\ldots$ such that for all $k\geq 0$
and for all $l\in[i_k,i_{k+1})$ and $l'\in[j_k,j_{k+1})$, 
$\sigma(l)=\sigma'(l')$.
%
A pointed trace $(\sigma',p')$ is a \emph{stuttering expansion} of
$(\sigma,p)$ if there is a sequence $p'=j_0<j_1<\ldots$ such that for
all $k\geq 0$ and for all $l\in [j_k,j_{k+1})$, $\sigma(p+k)=\sigma'(l)$.
%
%
We say that $\sigma$ is stuttering equivalent to $\sigma'$ if
$(\sigma,0)$ is stuttering equivalent to $(\sigma',0)$, and that
$\sigma'$ is a stuttering expansion of $\sigma$ if $(\sigma',0)$ is a
stuttering expansion of $(\sigma,0)$.

%
A {\em Kripke structure} is a tuple $\krip = \ktuple$, where $\States$
is a set of states, $\States_{\init} \subseteq \States$ is the set of
initial states, $\trans \subseteq \States \times \States$ is a
transition relation, and $L: S \rightarrow \alphabet$ is a labeling
function on the states of $\krip$. We require that for each
$\state \in \States$, there exists $\state' \in \States$, such that
$(\state, \state') \in \trans$.

A \emph{path} of a Kripke structure is an infinite sequence of states
$\state(0)\state(1)\cdots \in \States^\omega$, such that
$\state(0) \in \States_\init$ and
$(\state(i), \state({i+1})) \in \trans$, for all $i \geq 0$.
A \emph{trace} of a Kripke structure is a trace
$\trace(0)\trace(1)\trace(2) \cdots \in \alphabet^\omega$, such that
there exists a path $\state(0)\state(1)\cdots \in \States^\omega$ with
$\trace(i) = L(\state(i))$ for all $i\geq 0$.
Abusing notation we use $\trace=L(\rho)$ to denote that $\trace$ is
the trace corresponding to path $\rho$.
We denote by $\Trace(\krip, \state)$ the set of all traces of $\krip$
with paths that start in state $\state \in \States$,
We denote by $\Trace(\krip,A)$ the set of all traces that start from
some state in $A\subseteq S$ and $\Trace(\krip)$ as a short for
$\Trace(\krip,\States_{\init})$.

\subsubsection{HyperLTL.}

HyperLTL~\cite{clarkson14temporal} is a temporal logic that extends
LTL~\cite{pnueli77temporal,manna95temporal} for hyperproperties, which
allows reasoning about multiple execution traces simultaneously.
The syntax of HyperLTL is:
\begin{alignat*}{6}
  & \varphi   ::= \exists \pi. \varphi &&  \DefOR \forall \pi.\varphi && \DefOR \psi  && && \\
  & \psi  ::=  a_\pi && \DefOR \psi \lor \psi && \DefOR \neg \psi &&
   \DefOR \Next \psi && \DefOR \psi \U\psi
 \end{alignat*}
 where $\pi$ is a {\em trace variable} from an infinite supply of trace
 variables.
The intended meaning of $a_\pi$ is that proposition $a\in\alphabet$
holds in the current time in trace $\pi$.
Trace quantifiers $\exists\pi$ and $\forall\pi$ allow reasoning
simultaneously about different traces of the computation.
Atomic predicates $a_\pi$ refer to a single trace $\pi$.  Given a
HyperLTL formula $\varphi$, we use $\Vars(\varphi)$ for the set of
trace variables quantified in $\varphi$.
A formula $\varphi$ is well-formed if for all atoms $a_\pi$ in
$\varphi$, $\pi$ is quantified in $\varphi$ (\ie $\pi\in\Vars(\varphi)$)
and if no trace variable is quantified twice in $\varphi$.
Given a set of traces $T$, the semantics of a HyperLTL formula
$\varphi$ is defined in terms of trace assignments, which is a
(partial) map from trace variables to indexed traces
$\Pi:\Vars(\varphi)\partTo\PTR$.
The trace assignment with empty domain is denoted
by $\trass_\emptyset$.
We use $\Dom(\Pi)$ for the subset of $\Vars(\varphi)$ for which $\Pi$
is defined.
Given a trace assignment $\Pi$, a trace variable $\pi$, a trace
$\sigma$ and a pointer $p$, we denote by $\Pi[\pi \mapsto (\sigma,p)]$
the assignment that coincides with $\Pi$ for every trace variable
except for $\pi$, which is mapped to $(\sigma,p)$.
Also, we use $\Pi+n$ to denote the trace assignment $\Pi'$ such that 
$\Pi'(\pi)=\Pi(\pi)+n$ for all $\pi \in \Dom(\Pi) = \Dom (\Pi')$.
The semantics of HyperLTL is:
\newcommand{\modelsHLTL}{\ensuremath{\models}}
\[
  \begin{array}{@{}rl@{\hspace{2em}}c@{\hspace{2em}}l@{}}
    \Pi &\modelsHLTL_T \exists \pi. \varphi & \text{iff } & \text{for some 
$\sigma\in{}T$, } \Pi[\pi\mapsto (\sigma,0)]\modelsHLTL_T\varphi \\
    \Pi &\modelsHLTL_T \forall \pi. \varphi & \text{iff } & \text{for all 
                                                            $\sigma\in{}T$, } \Pi[\pi\mapsto (\sigma,0)]\models_T\varphi \\
    \Pi &\modelsHLTL_T \psi & \text{iff } & \Pi\modelsHLTL \psi \\
    \Pi &\modelsHLTL a_\pi & \text{iff } & a \in\sigma(p), \text{ where $(\sigma,p)=\Pi(\pi)$} \\
    \Pi &\modelsHLTL_T \psi_1 \Or \psi_2 & \text{iff } & \Pi\modelsHLTL_T \psi_1 \text{ or } \Pi\modelsHLTL_T \psi_2 \\
    \Pi &\modelsHLTL \neg \psi & \text{iff } & \Pi \not\modelsHLTL
                                               \psi \\
    \end{array}
    \]
    \[
      \begin{array}{@{}rl@{\hspace{2em}}c@{\hspace{2em}}l@{}}
        \Pi &\modelsHLTL \Next \psi & \text{iff } & (\Pi+1)\modelsHLTL \psi\\
    \Pi &\modelsHLTL \psi_1 \U \psi_2 & \text{iff } & \text{for some } j\geq 0\;\; (\Pi+j)\modelsHLTL \psi_2\\
 &&& \hspace{1em}\text{and for all $0\leq i<j$,} (\Pi+i)\models\psi_1
  \end{array}
\]
%
Note that quantifiers assign traces to trace variables and set the
pointer to the initial position $0$.
We say that a set of traces $T$ is a model of a HyperLTL formula
$\varphi$, denoted $T\models\varphi$ whenever $\trass_\emptyset\models_T\varphi$.
A Kripke structure $\krip$ is a model of a HyperLTL
formula $\varphi$, denoted by $\krip\models\varphi$, whenever
$\Trace(\krip)\models\varphi$.


%% file: ahltl.tex
\section{Asynchronous HyperLTL}
\label{sec:ahltl}

We introduce a temporal logic \AHLTL as an extension of HyperLTL to
express asynchronous hyperproperties.

\subsubsection{Trajectories.}
To model the asynchronous passage of time, we now introduce the notion
of a trajectory, which chooses when traces move and when they stutter.
Let $\V$ be a set of trace variables and let $I \subseteq \V$.
The $I$-successor of a trace assignment $\Pi$, denoted by $\Pi+I$, is
the trace assignment $\Pi'$ such that $\Pi'(\pi)=\Pi(\pi)+1$ if
$\pi\in I$ and $\Pi'(\pi)=\Pi(\pi)$ otherwise.
%
That is, the pointers of indices in $I$ advance by one step, while the
others remain the same.
A \emph{trajectory} $\traj: \traj(0)\traj(1)\traj(2)\cdots$ for a
formula $\varphi$ is an infinite sequence of non-empty subsets of
$\Vars(\varphi)$.
Essentially, in each step of the trajectory one or more of the traces
make progress.
A trajectory is fair for a trace variable $\pi \in \Vars(\varphi)$ if
there are infinitely many positions $j$ such that $\pi\in t(j)$.
A trajectory is fair it is fair for all trace variables in
$\Vars(\varphi)$.
%
%
Given a trajectory $\traj$, by $\traj^i$, we mean the suffix
$\traj(i)\traj(i+1)\cdots$.
Furthermore, for a set of trace variables $\V$, we use $\tjall_\V$ for
set of all trajectories for indices from $\V$.

\subsection{Syntax and Semantics of  Asynchronous HyperLTL}

The syntax of Asynchornous HyperLTL is:
\begin{alignat*}{4}
& \varphi ::= \exists \pi . \varphi && \mid \forall \pi. \varphi && \mid \Etau\psi \mid \Atau\psi \\
&\psi ::= a_\pi && \mid \lnot \psi && \mid \psi_1 \vee \psi_2 \mid 
\psi_1 \, \U \, \psi_2 \mid \Next \psi
\end{alignat*}
where $a \in \AP$, $\pi$ is a trace variable from an infinite supply
$\V$ of trace variables, $\Etau$ is the existential trajectory
modality and $\Atau$ is the universal trajectory modality.
The intended meaning of $\Etau$ is that there is a trajectory that
gives an intenrpretation of the relative passage of time between the
traces for which the temporal formula that relates the traces is
satisfied.
Dualy, $\Atau$ means that for all trajectories, the resulting
alignment makes the inner formula true.
It is important to note that there is no nesting of trajectory
modalities and that all temporal operators in a formula are
interpreted with respect to a single modality.

We use the usual syntactic sugar for Boolean operators
$\true \defAs a_\trvar \vee \neg a_\trvar$,
$\false \defAs \neg \true$,
$\varphi_1 \wedge \varphi_2 \defAs \neg(\neg \varphi_1 \vee \neg\varphi_2)$,
and the syntactic sugar for temporal operators
$\Event \varphi \defAs \true \U \varphi$,
$\varphi_1 \into \varphi_2 \defAs \neg \varphi_1 \vee \varphi_2$, and
$\G \varphi \defAs \neg \Event \neg \varphi$, etc.




As before, we use trace assignments for the semantics of \AHLTL.
%
%
%
Given $(\Pi,t)$ where $\Pi$ is a trace assignment and $t$ a trajectory,
we use $(\Pi,t)+1$ for the successor of $(\Pi,t)$ defined as
$(\Pi',t')$ where $t'=t^1$, and $\Pi'(\pi)=\Pi(\pi)+1$ if $\pi\in t(0)$
and $\Pi'(\pi)=\Pi(\pi)$ otherwise.
%
We use $(\Pi,t)+k$ as the $k$-th successor of $(\Pi,t)$.

The satisfaction of an asynchronous HyperLTL formula $\varphi$ over a
trace assignment $\trass$ and a set of traces $T$, denoted by
$\trass \models_T \varphi$ is defined as follows:
\[
  \begin{array}{rll@{\hspace{3em}}c@{\hspace{3em}}l}

 \trass  &\models_T& \exists \trvar. \varphi & \text{iff} & \text{for some } \trace 
\in T: \trass[\pi \mapsto (\trace, 0)] \models_T \varphi\\
    \trass  &\models_T& \forall \pi. \varphi & \text{iff} & \text{for all } \trace  \in T: \trass [\pi \mapsto (\trace, 0)] \models_T \varphi\\
    \trass  &\models_T& \Etau\psi & \text{iff} & \text{for some }\traj \in \tjall_{\Dom(\Pi)}.\, (\Pi,t) \models \psi\\
    \trass  &\models_T& \Atau\psi & \text{iff} & \text{for all }\traj \in \tjall_{\Dom(\Pi)}.\, (\Pi,t) \models \psi\\
(\trass,t) &\models& a_{\trvar} & \text{iff} & a \in \Pi(\trvar)(0)\\
(\trass,t) &\models& \neg \psi & \text{iff} & (\trass,t) \not 
\models\psi\\
(\trass,t) &\models& \psi_1 \, \vee \, \psi_2 & \text{iff} &
(\trass,t) \models \psi_1 \text{ or } (\trass,t) \models \psi_2 \\
    (\trass,t) &\models& \Next\psi & \text{iff} & (\trass,t)+1 \models \psi\\
    (\trass,t) &\models& \psi_1 \, \U \psi_2 & \text{iff} &
                                                            \text{for some } i\geq 0: (\trass,t)+i\models \psi_2 \, \text{ and } \, \\
    %
%
%
    &&&& \hspace{3em} \text{for all } j < i: (\trass,t)+j \models \psi_1\\
\end{array}
\]
We say that a set $T$ of traces satisfies a closed sentence $\varphi$,
denoted by $T \models \varphi$, if $\Pi_\emptyset \models_T \varphi$.
We say that a Kripke structure $\krip$ satisfies an \AHLTL formula
$\varphi$ (and write $\krip \models \varphi$) if and only if we have
$\Trace(\krip, \States_\init)\models \varphi$.

\subsection{Examples of A-HLTL}
\label{subsec:examples}

We illustrate the expressive power of \AHLTL by introducing the
asynchronous version of well-known properties.

%




\paragraph*{Linearizability}~\cite{hw90} requires that any history of execution of a 
concurrent data structure (i.e., sequence of invocation and response by 
different threads) matches some sequential order of invocations and responses:
\[
\lnz \DefinedAs \forall \pi.\exists \pi'.\Etau.\G(\history_\pi \Iff 
\history_{\pi'})\]
where $\mathsf{history}$ denotes method invocations (and not the
actual execution of the internal instructions of the concurrent
library) by the different threads and the response observed, trace
$\pi$ ranges over the concurrent data structure and $\pi'$ ranges over
its sequential counterpart.

\paragraph*{Goguen and Meseguer's noninterference (GMNI)}~\cite{gm82} stipulates
that, for all traces, the low-observable output must not change when
all high inputs are removed:
\[
\gmni =  \DefinedAs
\forall \pi. \exists \pi'. \Etau.
(\G \lambda_{\pi'}) \, \And \, \G(\lo_\pi \Iff \lo_{\pi'})
\]
where $\lambda_{\pi'}$ expresses that all of the high inputs in the
current state of $\pi'$ have dummy value $\lambda$, and denotes
low-observable output proposition.

\paragraph*{Not never terminates}~\cite{l80} requires that for every
initial state, there is a terminating trace and a non-terminating trace:
\[
\nnt \DefinedAs \forall \pi. \exists \pi'.\exists \pi''.\Etau.
(\pi[0] = \pi'[0] = \pi''[0]) \; \Into \;
(\F \, \term_{\pi'} \; \And \; \G \neg \term_{\pi''})
\]



\paragraph*{Termination-insensitive noninterference}~\cite{ss01} requires that 
for two executions that start from a low-observable states, information leaks 
are permitted if they are transmitted purely by the program's termination 
behavior. That is, the program may diverge on some high inputs and terminate on 
others:

\[\begin{array}{l}
	\tin \DefinedAs \forall \pi. \forall \pi'. \Etau.  
	\Big(l_\pi \Iff l_{\pi'}\Big) \; \Into \;
    \begin{pmatrix}
      \begin{array}{l}
	(\G \neg \term_\pi \Or  \G \neg \term_{\pi'})  \Or \\
        \F(\term_\pi \And \term_{\pi'} \And l_\pi \Iff l_{\pi'})
      \end{array}
    \end{pmatrix}
\end{array} \]
\label{form:TNI1}


\paragraph*{Termination-sensitive noninterference}~\cite{ahss08} 
Termination-sensitive noninterference is the same as termination insensitive, 
except that it forbids one trace to diverge and the other to terminate:
\[\begin{array}{l}
	\tsn \DefinedAs \forall \pi. \forall \pi'. \Etau. 
	\Big(l_\pi \Iff l_{\pi'}\Big) \Into 
    \begin{pmatrix}
      \begin{array}{l}
	(\G \neg \term_\pi \And  \G \neg \term_{\pi'})  \Or \\
        \F(\term_\pi \And \term_{\pi'} \And l_\pi \Iff l_{\pi'})
      \end{array}
    \end{pmatrix}
    
\end{array}
\]



%% file: algorithm.tex
\section{Model-Checking \AHLTL}
\label{sec:decidable}

In this section, we show the decidability of the model-checking problem
for two classes of \AHLTL formulas using two different algorithms:
\begin{compactenum}[(1)]
\item a \emph{stuttering} construction in which we modify the Kripke
  structure $\krip$ to accept all stuttering expansions of paths in
  $\krip$; and
\item an \emph{acceleration} construction in which the modified Kripke
  structure accelerates jumping directly to the synchronization
  points.
\end{compactenum}
In both cases the problem is reduced to model-checking HyperLTL
formulas, which is known to be
decidable~\cite{finkbeiner13temporal,clarkson14temporal}.
We describe each construction separately.

\subsection{The Stuttering Construction}
\label{subsec:stuttering}

We consider first \AHLTL formulas of the form
\(
  \forall \pi_1\ldots\pi_n.\Etau.\psi.
\)
We will then extend our results to the $\exists^*$ fragment, to handle
the $\Atau$ trajectory modality and to a larger collection of
predicates.
The class of temporal formulas $\psi$ that we handle are called
\emph{admissible} formulas, and are defined as the Boolean combination
of:
\begin{compactenum}
\item any number of state formulas, which may relate propositions $p_{\pi_i}$ of different traces 
arbitrarily;
  
\item any number temporal formulas (called \emph{monadic
temporal formulas}), each of which only uses one trace
  variable and is invariant under stuttering (guaranteed for example
  by forbidding the use of $\Next$), and
  
\item one \emph{phase formula}, which is an invariant that can relate
  different traces in a restricted way (see below).
\end{compactenum}
Given an admissible formula $\psi$, we use $\psiPH$ for its phase
formula, and we use $\psi[\psiPH\subst\xi]$ for the formula that results from
$\psi$ by replacing $\psiPH$ with $\xi$.
Since $\psiPH$ occurs only once in $\psi$, we use the fact that
$\psiPH$ appears with a single polarity.
We present here the construction for positive polarity which is the
case in all practical formulas (the case for negative polarity is
analogous).
%


The algorithm has two parts.
First, we generate the \emph{stuttering} Kripke structure $\kripST$
whose paths are the stuttering expansions of paths in the original
Kripke structure $\krip$.
Then, we modify the admissible formula $\psi$ into $\psiSync$ such
that $\krip\models \forall\pi_1\ldots\pi_n.\Etau.\psi$ if and only if
$\kripST\models \forall\pi_1\ldots\pi_n.\psiSync$.
We describe each of the concepts separately.

\paragraph{Phase formulas.}
We first define \emph{atomic phase formulas}
\( (\bigwedge_{p\in P} p_{\pi_i}\Iff p_{\pi_j}) \)
which are characterized by $(\pi_i,\pi_j,P)$, where $P\subseteq \AP$
and $\pi_i$ and $\pi_j$ are two different trace variables.
We use \emph{color} to refer to a valuation of the variables in $P$.
Essentially, an atomic phase formula asserts that all propositions in
$P$ coincide in both traces at all points in time, that is, both
traces exhibit the same sequence of colors.
Since the passage of time proceeds at different speeds in the
different traces---according to the trajectory--- atomic phase formulas
state the traces for $\pi_i$ and $\pi_j$ are sequences of phases of
the same color, where corresponding phases may have different lengths.
A phase formula is formed from atomic formulas as follows:
\[ \Always \big( \bigwedge_{p\in P^1} p_{\pi^1_i}\Iff p_{\pi^1_j} \And \cdots \And \bigwedge_{p\in P^k} p_{\pi^k_i}\Iff p_{\pi^k_j}\big) \]
We use $\PS:\{(\pi_i^1,\pi_j^1,P^1),\ldots,(\pi_i^k,\pi_j^k,P^k)\}$
for the collection of predicates and trace variables that characterize
a phase formula.

\paragraph{Stuttering Kripke structure.}
We start from $\krip$ and create $\kripST$ that accepts the stuttering
expansions of traces in $\krip$.
First, the alphabet of atomic propositions is enriched with a fresh
proposition $\ST$, that is $\AP^{\ST}=\AP\cup\{\ST\}$, to encode whether
the state represents a real move or a stuttering move.
Given $\krip=\tupleof{S, \States_\init, \trans, L}$, the stuttering Kripke
structure is $\kripST=\tupleof{\SST, \States_\init, \transST,\LST}$ where:
\begin{compactitem}
\item $\SST=S\cup\{s^\ST \mid s\in S\}$ contains two copies of each state
  in $S$, where we use $s^\ST$ to denote the stuttering state that
  corresponds to $s$;
  
\item
  $\transST=\trans\cup \{(s,s^\ST)\}\cup\{(s^\ST,s^\ST)\}
  \cup\{(s^\ST,s') \mid \text{ for every } (s,s')\in\trans\}$.
\item $\LST(s)=L(s)$ for $s\in S$, and $\LST(s^\ST)=L(s)\cup\{\ST\}$.
\end{compactitem}
The construction generates a Kripke structure which is linear in the
size of the original Kripke structure.
It is easy to see that every stuttering expansion of a path of $\krip$
has a corresponding path in $\kripST$, where the repeated version of
state $s$ is captured by state $s^\ST$.
Conversely every path $\rho'$ in $\kripST$ whose trace satisfies
$\Always\Event \neg\ST$ can be turned into its ``stuttering
compression'' by removing all stuttering states, which is a path of
$\krip$.
Note that the constraint $\Always\Event \neg\ST$ guarantees that there
are infinitely many non-stuttering positions in $\rho'$, so $\rho$ is
well-defined.
This constructions makes a one-to-one correspondence between a
trajectory and a tuple of traces of $\krip$, with a corresponding
tuple of traces of $\kripST$.

\paragraph{State and monadic formulas are not affected by trajectories.}
State formulas are relational formulas that are evaluated at the
beginning of the computation.
Temporal monadic formulas only refer to one trace variable and are
stuttering invariant by definition.
Therefore, none of these formulas are affected by the stuttering induced
by a trajectory, as the relative stuttering among traces does not
affect their truth valuation.
We first note that given a trace assigned for each of the trace
variables in $\Vars(\varphi)$ the truth value of state formulas and
monadic formulas does not depend on the trajectory chosen.
%



\paragraph{Phase alignment of asynchronous sequences.}
We use the stuttering in $\kripST$ to encode the relative progress of
traces as dictated by a trajectory.
We will now introduce synchronous HyperLTL formulas to reason in
$\kripST$ about the corresponding states during the asynchronous
evaluation in $\krip$.
The important concept is that of ``phase changes'', which are the
points in a trace $\sigma$ at which the valuation of the predicates
$P$ in an atomic phase formula $(\pi_i,\pi_j,P)$ change.
Let $\Pi$ be a trace assignment for traces in $\krip$ that maps $\pi_i$
to a pointed trace $(\sigma,l)$.
We say that in assignment $\Pi$, trace variable $\pi_i$ is \emph{about
  to change phase} with respect to $(\pi_i,\pi_j,P)$ if for some
$p\in P$ either $p\in\sigma(l)$ but $p\notin\sigma(l+1)$ or
$p\notin\sigma(l)$ but $p\in\sigma(l+1)$.
Note that in $\kripST$ the next relevant letter (the one
corresponding to $\sigma(l+1)$ is the first letter that is not a
stuttering letter).
Formula $\Change_P(\pi_i)$ captures that the next non-stuttering step
of $\pi_i$ is a phase change (with respect to predicates in $P$ and
therefore with respect to atomic phase formula $\alpha$):
\[
  \Change_P(\pi_i)\DefinedAs \bigvee_{p\in P}p_{\pi_i}\not\Iff\Next(\ST_{\pi_i} \U p_{\pi_i})
\]
A phase change for $\pi_i$ in atomic phase formula $(\pi_i,\pi_j,P)$
implies that $\pi_j$ must also proceed to change phase.
The second observation is that when $\pi_i$ and $\pi_j$ are not
changing phases, any choice that the trajectory makes will preserve
the valuation of the atomic phase formula.

We now capture formally this intuition as formulas.
%
%
Predicate $\Move(\pi_i)\DefinedAs\Next(\neg\ST_{\pi_i})$ indicates
whether trace variable $\pi_i$ will move (and not stutter) at a given
instant of the computation.
The following temporal formula captures the consistency criteria of
phase changes as a synchronized decision for moving traces $\pi_i$ and
$\pi_j$ related by an atomic phase formula $(\pi_i,\pi_j,P)$:
\[
  \begin{array}{l}
    \Align_{(\pi_i,\pi_j,P)}\DefinedAs\\
    \hspace{5em}\begin{array}{l}
    \begin{pmatrix}
    \begin{array}{rcrcll}
      (\Move(\pi_i)&\And&\Move(\pi_j)) & \Then & (\Change_P(\pi_i)\Iff \Change_P(\pi_j)) & \And \\
      (\Move(\pi_i)&\And&\neg\Move(\pi_j)) & \Then & \neg\Change_P(\pi_i) & \And \\
      (\neg\Move(\pi_i)&\And&\Move(\pi_j)) & \Then & \neg\Change_P(\pi_j) & \\
    \end{array}
  \end{pmatrix}
                \end{array}
  \end{array}
\]
We will reduce the model-checking problem in \AHLTL to checking in
$\kripST$ that tuples of traces that align phase changes---for all
atomic phase formulas--- satisfy all sub-formulas of the specification
$\psi$.
The following two formulas express that all atomic phase formulas
align, and that all traces are fair (all traces eventually move):
\[
  \Phases\DefinedAs\bigwedge_{(\pi_i,\pi_j,P)\in\PS}
  \Align_{(\pi_i,\pi_j,P)}\hspace{7em}
  \Fair\DefinedAs\bigwedge_{\pi_i\in\{\pi_1\ldots\pi_n\}}\Always\Event\neg\ST_i
\]
We will then check in $\kripST$ that all stuttering traces that align
phases and are fair satisfy the desired formula $\psi$, that is
$(\Always\Phases\And\Fair)\Then\psi$.
Note that all those tuples of traces that do not align phases are
ruled out in the antecedent.

A final technical detail in the construction is that we must guarantee
that for all tuples of paths of $\krip$ there are stuttering
expansions that are fair and align phases, and that they have the same
number of phases.
Otherwise, there are paths of $\krip$ that cannot be aligned, which
inevitably leads to a violation of $\psiPH$.
It could be the case that some tuple of traces of $\krip$ cannot
possibly align the phase changes corresponding to all atomic phase
formulas.
This can happen in two cases: (1) when two traces have
different number of phases, and (2) when there is a circular
dependency between the atomic formulas that force the trajectory to
synchronize the traces in incompatible orders.
%
%
The first case is captured by:
\[
  \Bad \DefinedAs \bigvee_{(\pi_i,\pi_j,P)} \Big(\Always\neg\Change_P(\pi_i)\Big) \not\Iff 
  \Big(\Always\neg\Change_P(\pi_j)\Big)
\]
%
%
The second case is captured by the following formula, where
$\Cycles(\psiPH)$ are the sequences of atomic formulas that form a
simple cycle, that is \linebreak
$[(\pi^0,\pi^1,P^0),(\pi^1,\pi^2,P^1)\ldots(\pi^k,\pi^0,P^k)]$ such
that the second trace variable is the first trace variable of the next
atomic phase formula, circularly (see Ex.~$1$ below):
 \[
   \Block\DefinedAs\bigvee_{C\in\Cycles(\psiPH)}
   \Big(\bigwedge_{(\pi_i,\pi_j,P)\in C} \Change_P(\pi_i)
   \And\neg\Change_P(\pi_j)\Big)
 \]
 Essentially, $\Block$ encodes whether the set of traces involved
 cannot proceed without violating $\Phases$, because $\Align$ forbids
 all traces involved to move.
 Hence, the formula $\Phases\U(\Bad\Or\Block)$ captures to those
 traces of $\kripST$ that contain an aligned prefix of computation
 that lead to a miss-alignment or a block.
 The proof of correctness shows that given a tuple of traces of
 $\krip$, if there is a trajectory that aligns the phase changes
 (which must exist if there is a trajectory that makes $\psiPH$ true),
 then all trajectories that respect $\Always\Phases$ will also align
 the phase changes (and also satisfy $\psiPH$).

%
%

 We are finally ready to describe the synchronous phase formula
 $\psiSync$.
 First, this formula is only evaluated against tuples of fair traces,
 which correspond to the stuttering extensions of paths of $\krip$.
 Then, the phase formula $\psiPH$ is translated into a formula that
 captures (1) that following a phase alignment cannot lead to a block
 or to two traces changing phases a different number of times, and (2)
 that if phases are aligned then $\psiPH$ holds.
 Formally,
 \[
   \psiSync\DefinedAs \Fair\Then\psi[\psiPH\subst\psi'], \;\;\;\;\text{ where } \psi'=
   \begin{pmatrix}
    \begin{array}{l}
       \neg(\Phases\;\U(\Bad\Or\Block)) \;\And \\
       \Always\Phases\Then\psiPH
     \end{array}
   \end{pmatrix}
\]
%

\medskip
\stepcounter{example}
\noindent \textit{Example \theexample.}
\begin{figure}[b!]
  \vspace{-1em}
  \begin{tabular}{l@{\hspace{1em}}l@{\hspace{1em}}l}
  \includegraphics[scale=0.45]{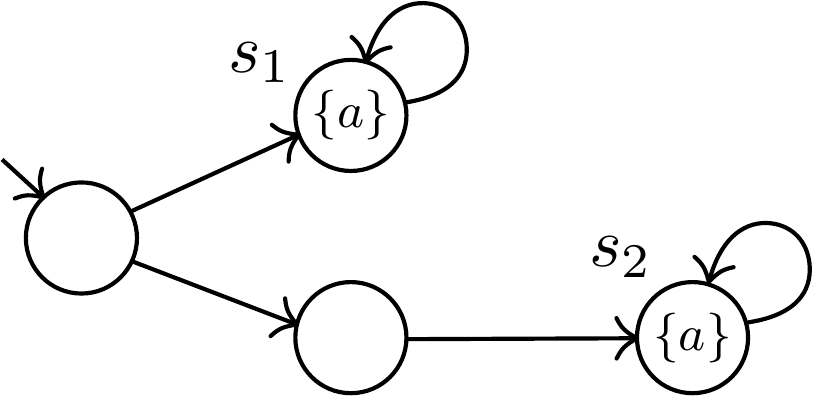} &
  \includegraphics[scale=0.45]{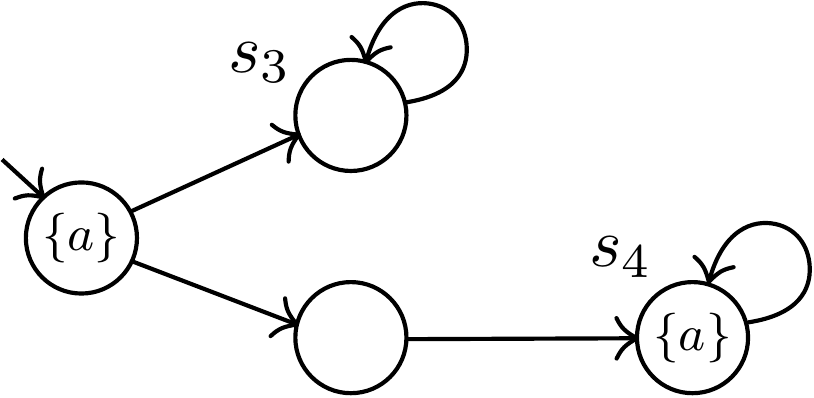} &
   \includegraphics[scale=0.45]{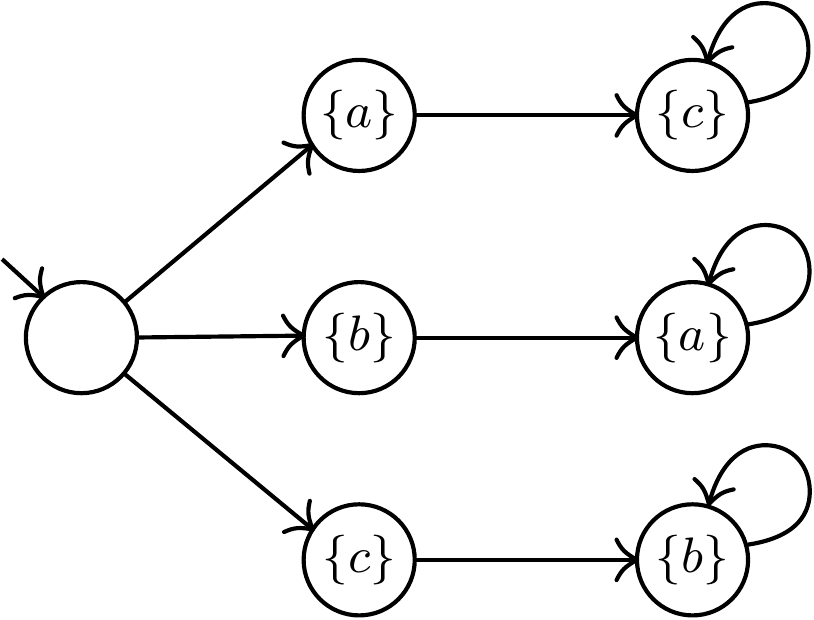}
  \end{tabular}
  \caption{Kripke structure $\krip_1$ (left), $\krip_2$ (middle) and
    $\krip_3$ (right).}
  \label{fig:kthree}
\end{figure}
We illustrate the previous definitions with the Kripke structures
$\krip_1$, $\krip_2$ and $\krip_3$ in Fig.~\ref{fig:kthree} and their
stuttering variants $\krip_1^\ST$, $\krip_2^\ST$ and $\krip_3^\ST$
Consider formula
$\forall\pi_1.\forall\pi_2.\Etau.\Always(a_{\pi_1}\Iff a_{\pi_2})$.
Consider the following trace assignments:

\vspace{0.5em}
\noindent\begin{tabular}{|l|l|l|}\hline && \\[-0.7em]
    \begin{array}{rrllll}
      \Pi^1(\pi_1)&\mapsto& \{\} & \{\ST\} & \{a\} & \ldots \\
      \Pi^1(\pi_2)&\mapsto& \{\} & \{\}    & \{a\} & \ldots
    \end{array} &
    \begin{array}{rrllll}
      \Pi^2(\pi_1)&\mapsto& \{\} & \{a\} & \{a\} &  \ldots \\
      \Pi^2(\pi_2)&\mapsto& \{\} & \{\} & \{a\}  & \ldots
    \end{array} &
    \begin{array}{rrllll}
      \Pi^3(\pi_1)&\mapsto& \{a\} & \{\} & \{\} & \ldots \\
      \Pi^3(\pi_2)&\mapsto& \{a\} & \{\} & \{a\} & \ldots
    \end{array}\\[1em]\hline
  \end{tabular}
\vspace{0.5em}

\noindent Consider the trace assignment $\Pi^1$ on the left, where
$\pi_1$ is a trace of $\krip_1^\ST$ corresponding to the
path of $\krip_1$ that visits $s_1$, and $\pi_2$ corresponds to the
path that visits $s_2$. This trace assignment aligns the atomic phase
formula $(\pi_1,\pi_2,\{a\})$ at all positions. In particular, at
position $0$, we have $\Change_{\{a\}}(\pi_1)$, but
$\neg\Change_{\{a\}}(\pi_2)$, and $\neg\Move(\pi_1)$ and
$\Move(\pi_2)$, as $\Align_{\{a\}}$ requires.

Consider now the trace assignment $\Pi^2$ in the middle, where again
$\pi_1$ corresponds to the path in $\krip_1^\ST$ that visits $s_1$ and
$\pi_2$ the path that visits $s_2$.
In this case, we have $\neg\Align_{\{a\}}$ at position $0$ because
$\Change_{\{a\}}(\pi_1)$ and $\neg\Change_{\{a\}}(\pi_2)$ hold, and
both $\Move(\pi_1)$ and $\Move(\pi_2)$.
Consider now $\Pi^3$ on the right, where $\pi_1$ corresponds to the
path of $\krip_2^\ST$ that visits $s_3$ and $\pi_2$ to the path of
$\krip_2^\ST$ that visits $s_4$.
  In this case $\Align_{\{a\}}$ holds at $0$ and $\Bad$ holds at $1$
  because at $1$,  $\Always\neg\Change_{\{a\}}(\pi_1)$ holds, but not 
  $\Always\neg\Change_{\{a\}}(\pi_2)$.
  %
  %
  Therefore, $\Phases\U(\Bad\Or\Block)$ holds for $\Pi^3$.
  Finally, consider
  $\forall\pi_1.\forall\pi_2.\forall\pi_3.\Etau.\Always( a_{\pi_1}\Iff
  a_{\pi_2} \And b_{\pi_2}\Iff b_{\pi_3} \And c_{\pi_3}\Iff c_{\pi_2}
  )$
  and the trace assignment $\Pi$ of $\krip_3^\ST$ shown below on the
  \begin{wrapfigure}[5]{l}{0.35\textwidth}
    \vspace{-1.6em}
    \begin{tabular}{|l|}\hline
      \begin{array}{rrlllll}
        \Pi(\pi_1)&\mapsto& \{\} & \{\} & \{a\} & \{c\}\ldots \\
        \Pi(\pi_2)&\mapsto& \{\} & \{\} & \{b\} & \{a\}\ldots \\
        \Pi(\pi_3)&\mapsto& \{\} & \{\} & \{c\} & \{b\}\ldots \\
      \end{array}\\[1em]\hline
    \end{tabular}
  \end{wrapfigure}
  left.
  In this case $\Phases$ holds at position $0$ and $\Block$ holds at
  position $1$.
  This is
 because $\Change_{\{a\}}(\pi_1)$ and $\neg\Change_{\{a\}}(\pi_2)$,
$\Change_{\{b\}}(\pi_2)$ and $\neg\Change_{\{b\}}(\pi_3)$, and
$\Change_{\{c\}}(\pi_3)$ and also $\neg\Change_{\{c\}}(\pi_1)$.
This illustrates that it will not be possible to align all three
atomic phase formulas.

\medskip


We are now ready to state the main result of this section.

\newcounter{thm-sync-async}
\setcounter{thm-sync-async}{\value{theorem}}

\begin{theorem}
  \label{thm:sync-async}  
  Let $\krip$ be a Kripke structure and $\psi$ an admissible formula.
  Then, $\krip\models \forall\pi_1\ldots\pi_n.\Etau.\psi$ if
  and only if $\kripST\models \forall\pi_1\ldots\pi_n.\psiSync$.
\end{theorem}



\noindent 
Dually, to show that the $\exists^*$ fragment is decidable, we
consider replacing $\psiPH$ by the formula
\[ \psiESync\DefinedAs \Fair \And
  \psi[\psiPH\subst(\Always\Phases \And \psiPH)] \]

\newcounter{thm-EE}
\setcounter{thm-EE}{\value{theorem}}

\begin{theorem}
  \label{thm:EE}
  Let $\krip$ be a Kripke structure and $\psi$ an admissible formula.
  Then $\krip\models \exists\pi_1\ldots\pi_n.\Etau.\psi$ if and only if
  $\kripST\models \exists\pi_1\ldots\pi_n.\psiESync$.
\end{theorem}

The proof of Theorem~\ref{thm:EE} takes a witness tuple and trajectory
in $\krip$ and shows that the induced tuple in $\kripST$ is $\Fair$,
satisfies $\Always\Phases$ and that the valuation of $\psiPH$ is
preserved.
Similarly, as before, tuples of traces of $\kripST$ that are fair and
follow phase alignments induce a trajectory on their stuttering
compression that also preserve $\psiPH$.

\begin{corollary}
  The problems of model-checking $\forall^*$ admissible \AHLTL
  formulas and $\exists^*$ admissible \AHLTL formulas is decidable.
\end{corollary}

We finally consider the negation of phase formulas, called
\emph{co-phase formulas}, which are formulas of the form
$\Event \neg R$ where $R$ a conjunction of atomic phase formulas.
Interestingly, deciding co-admissible formulas (consisting of Boolean
combinations of state-formulas, monadic temporal formulas and one
co-phase formula in positive polarity) is easier than before, as one
can turn the co-phase formula into a monadic formula enumerating all
the violations of the atomic phase formulas ($p\in P$ such that
$p_{\pi_i}\not\Iff p_{\pi_j}$) turns the atomic phase formula into
$(\Event p_{\pi_i} \And \Event \neg p_{\pi_j}\big) \Or (\Event \neg
p_{\pi_i} \And \Event p_{\pi_j}\big)$.
It follows that model-checking co-admissible formulas is also
decidable (for both $\forall^*$ and $\exists^*$).
Note that an admissible formula in negative polarity is a
co-admissible formula in positive polarity (and vice versa).
Finally, since
$\krip\models \forall \pi_1\ldots\forall\pi_n.\Atau.\psi$ if and only
if $\krip\not\models \exists \pi_1\ldots\exists\pi_n.\Etau.\neg\psi$,
it follows that model-checking is also decidable for the $\Atau$
modality for both admissible and co-admissible formulas (in both
polarities), and for both the $\forall^*$ and $\exists^*$ fragments.

\newcounter{thm-adm-coadm}
\setcounter{thm-adm-coadm}{\value{theorem}}

\begin{theorem}
  \label{thm:adm-coadm}
  Model-checking $\forall^*$ or $\exists^*$ admissible and
  co-admissible formulas is decidable both for formulas with $\Etau$
  and formulas with $\Atau$.
\end{theorem}

\newcommand{\Q}{\mathbb{Q}}

\subsection{The Accelerating Construction}
\label{subsec:accelerating}

The admissible formula in the stuttering construction can express many
formulas of interest, but the quantifier structure admits no
quantifier alternation.
We now consider a second decidable fragment for \AHLTL formulas
consisting of formulas with arbitrary quantification
$ \Q_1 \pi_1. \Q_2\pi_2.\ldots. \Q_n\pi_n \Etau. \psi$
such that $\Q_i \in \{\forall, \exists\}$, but where $\psi$ is an admissible formula where all atomic phase
formulas use the same atomic predicates $P\subseteq\AP$.
We call these admissible formulas \emph{simple admissible formulas}.
The proof of decidability proceeds this time by creating the
\emph{accelerated} Kripke structure $\kripAcc$, where paths jump in
one step to the next phase change, and reducing to a HyperLTL
model-checking problem on $\kripAcc$.

\paragraph{Accelerated Kripke structure.}
The main idea of the acceleration construction is to convert a finite
sequence of transitions in $\krip$ that only change phase in the last
transition into a single transition in $\kripAcc$.
Also, an infinite sequence of transitions with no phase change is
transformed into a self-loop around a sink state.
The alphabet remains the same, $\AP$.
Given $\krip=\tupleof{S, \States_\init, \trans, L}$, the accelerated Kripke
structure is $\kripAcc=\tupleof{\SAcc, \States_\init, \transAcc,\LAcc}$ where:
\begin{compactitem}
\item $\SAcc=S\cup\{\sbot \mid s\in S\}$ contains two copies of each state
  in $S$, where we use $\sbot$ to denote the sink state associated
  with $s$.
  We use $\Color(s)$ for the phase of $s$, that is, the concrete
  valuation in $s$ of the Boolean predicates in $P$ of the atomic
  phase formula.
  
\item For every states $s, s' \in S$ such that
  $\Color(s)\neq\Color(s')$, if there is a finite path
  $ss_2s_3\ldots s_ns'$ in $\krip$ such that
  $\Color(s)=\Color(s_2)=\cdots=\Color(s_n),$
  then we add a transition $(s,s')$ to $\transAcc$.
  These transitions model the jump at the frontier of phase changes.
  Additionally, if $s$ can be a sink we add a transition $(s,s_\bot)$
  and a self-loop from $s_\bot$ to itself.
\item $\LAcc(s)=L(s)$ for $s\in S$, and $\LAcc(\sbot)=L(s)$.
\end{compactitem}

This construction can, with standard techniques, be enriched to encode
the satisfaction of the temporal monadic formulas along paths of
$\krip$, and then also accelerate the fairness conditions (annotating
the accepting states reached along the accelerated paths) into
$\kripAcc$.

\paragraph{Relating paths to accelerated paths.}
We now define two auxiliary functions to aid in the proof.
\begin{itemize} 
\item 
  The first function, $\acc$, maps paths in $\krip$ into paths in
  $\kripAcc$.
Let $s$ be an arbitrary state of $\krip$ and
$\rho:ss_1s_2s_3\ldots$ an outgoing path from $s$.
Either there are infinitely many phase changes in $\rho$ or only
finitely many changes.
We create the path $\rho'=\acc(\rho)$ as follows. 
The initial state of $\rho$, that is, $s$, is preserved.
The states $s_{i_j}$ in $\sigma$ that are color changes (that is
$\Color(s_{i_j-1} )\neq\Color(s_{i_j}$) are also preserved, while the
states $s_k$ with $\Color(s_{k-1})=\Color(s_k)$ are removed from
$\rho$.
If there are only finitely many color changes in $\rho$, with $r$
being the last state preserved, then we pad the path with
$r_\bot^\omega$, so $\rho'$ is also an infinite path.
It is easy to see that $\rho'$ is a path of $\kripAcc$ outgoing $s$.
It is also easy to see that the phase changes in $\rho$ and
$\rho'$ are the same.

\item 
  The second map, $\dec$, takes a path $\rho':ss_1's_2'\ldots$ of $\kripAcc$
  and maps it to a path of $\krip$ as follows.
  For every transition $(s_i',s_{i+1}')$ in $\rho$ such that $s_{i+1}'$
  is not of the form $r_\bot$, there is a finite path $r_1r_2\ldots r_m$
  in $\krip$ from $s_i'$ into $s_{i+1}'$ that visits only states with the
  same color as $s_i'$, except $s_{i+1}$ that is a color change.
  In $\rho$, we insert $r_1r_2\ldots r_m$ between $s_i'$ and $s_{i+1}'$.
  Now, if for some $j$, $s'_j$ is of the form $r_\bot$ then
  $s'_k=r_\bot$ for all $k>j$.
  In $\krip$ there must an infinite path from $s'_j$ that only visits the
  same color as $s'_j$. 
  We remove all successor states after the first such $r_\bot$ state
  and replace it with one such infinite path.
\end{itemize}

Given a trace assignment $\Pi$ for formula
$\Q_1\pi_1.\ldots \Q_n\pi_n.\Etau.\psi$ that assigns
$\Pi(\pi_i)=(\sigma_i,0)$ for every $i$ and a path assignment $\Pi'$
for formula $\Q_1\pi_1.\ldots.\Q_n\pi_n.\psi$ that assigns
$\Pi'(\pi_)i=(\sigma_i',0)$, we write $\acc(\Pi)=\Pi'$ if the paths
that generate the corresponding traces are related by $\acc$.
Similarly we defined $\dec(\Pi')=\Pi$.
%
%
It is easy to show from the construction above that if
$\Pi\models\Etau\psi$ then $\acc(\Pi)\models\psi$, and if
$\Pi'\models\psi$ then $\dec(\Pi')\models\Etau\psi$.


The main result for the accelerating construction follows immediately
from this observation and allows to reduce the model-checking problem
to HyperLTL.

\newcounter{thm-accelerating}
\setcounter{thm-accelerating}{\value{theorem}}

\begin{theorem}
  \label{thm:accelerating}
  Let $\krip$ be an arbitrary Kripke structure,
  $\Q_1\pi_1.\ldots.\Q_n\pi_n.\Etau.\psi$ such that $\psi$ is a simple
  admissible formula.
  Then $\krip\models{}\Q_1\pi_1.\ldots{}\Q_n\pi_n\Etau.\psi$ if and only if
  $\kripAcc\models{}\Q_1\pi_1.\ldots{}\Q_n\pi_n.\psi$.
\end{theorem}

\subsection{Decidable Practical  \AHLTL formulas}
\label{subsec:dec-examples}

We revisit the properties expressed in Section~\ref{subsec:examples}.
\begin{itemize}
\item \textit{Linearizability}. The property $\lnz$ is of the form
  $\forall \pi.\exists \pi'.\Etau.\G(\history_\pi \Iff
  \history_{\pi'})$ where the temporal formula is a simple admissible
  formula.
  Therefore $\lnz$ is decidable by the accelerating construction.
\item \textit{Goguen and Meseguer's non-interference}. The property
  $\gmni$ is expressed by 
  $\forall \pi. \exists \pi'. \Etau.  (\G \lambda_{\pi'}) \, \And \,
  \G(\lo_\pi \Iff \lo_{\pi'})$, that is, a Boolean combination of a
  monadic temporal formula and a simple admissible formula. Therefore,
  $\gmni$ is decidable by the acceleration algorithm.
\item \textit{Not never terminates}. Formula $\nnt$ is simply a
  Boolean combination of state formulas and monadic temporal formulas:
  $\forall \pi. \exists \pi'.\exists \pi''.\Etau.  (\pi[0] = \pi'[0] =
  \pi''[0]) \; \Into \; (\F \, \term_{\pi'} \; \And \; \G \neg
  \term_{\pi''})$, so it is again decidable by the acceleration
  construction.
\item \textit{Termination-insensitive noninterference}. To handle
  $\tin$ we rewrite the formula as follows
  \[\begin{array}{l}
	\tin \DefinedAs \forall \pi. \forall \pi'. \Etau.  
      \Big(l_\pi \Iff l_{\pi'}\Big) \; \Into \;
    \begin{pmatrix}
      \begin{array}{l}
	(\G \neg \term_\pi \Or  \G \neg \term_{\pi'})  \Or \\
        \G\big((l_\pi\And \term_{\pi}) \Iff(l_{\pi'}\And \term_{\pi'})\big)
      \end{array}
    \end{pmatrix}
    \end{array} \]
\label{form:TNI2}
  Note that $(l_\pi\And \term_{\pi})$ can be turned into a state
  predicate of $\pi$.  This formula is equivalent because the last
  case is evaluates precisely to $l_\pi\Iff l_{\pi'}$ when both traces
  terminate.
  This formula can be handled by the stuttering construction.
\item  \textit{Termination-sensitive noninterference}. Similarly, to handle
  $\tsn$ we rewrite the formula as 
  \[\begin{array}{l}
	\tsn \DefinedAs \forall \pi. \forall \pi'. \Etau.  
      \Big(l_\pi \Iff l_{\pi'}\Big) \; \Into \;
    \begin{pmatrix}
      \begin{array}{l}
	(\G \neg \term_\pi \And  \G \neg \term_{\pi'})  \Or \\
        \G\big((l_\pi\And \term_{\pi}) \Iff(l_{\pi'}\And \term_{\pi'})\big)
      \end{array}
    \end{pmatrix}
    \end{array} \]
  This is again equivalent because the last case again is the only
  relevant case when both paths terminate.
  Again, this case is covered by the stuttering construction.
\end{itemize}


%% file: colors.tex


%% file: undec.tex
\section{Undecidability and Lower-bound Complexity}
\label{sec:undec}

In this section, we show that the general problem of model-checking
\AHLTL is undecidable.
Then, we show a polynomial reduction from the synchronous HyperLTL
model-checking into \AHLTL model-checking, which shows that even for
those \AHLTL formulas for which the model-checking is decidable, this
problem is no easier than the corresponding problem for HyperLTL,
which is known to be \comp{PSPACE-hard} in the size of the Kripke
structure.

\newcounter{thm-undec}
\setcounter{thm-undec}{\value{theorem}}

\begin{theorem}
  \label{thm:undec}
  Let $\krip$ be a Kripke structure and $\varphi$ be an asynchronous
  HyperLTL formula. The problem of determining whether or not
  $\krip \models \varphi$ is undecidable.
\end{theorem}

\begin{proof}[sketch] We reduce the complement of the {\em post
    correspondence problem
    (PCP)}~\cite{post47recursive,sipser12introduction} to the \AHLTL
  model checking problem.
\begin{figure}[b!]
	\centering
	\includegraphics[scale=.8]{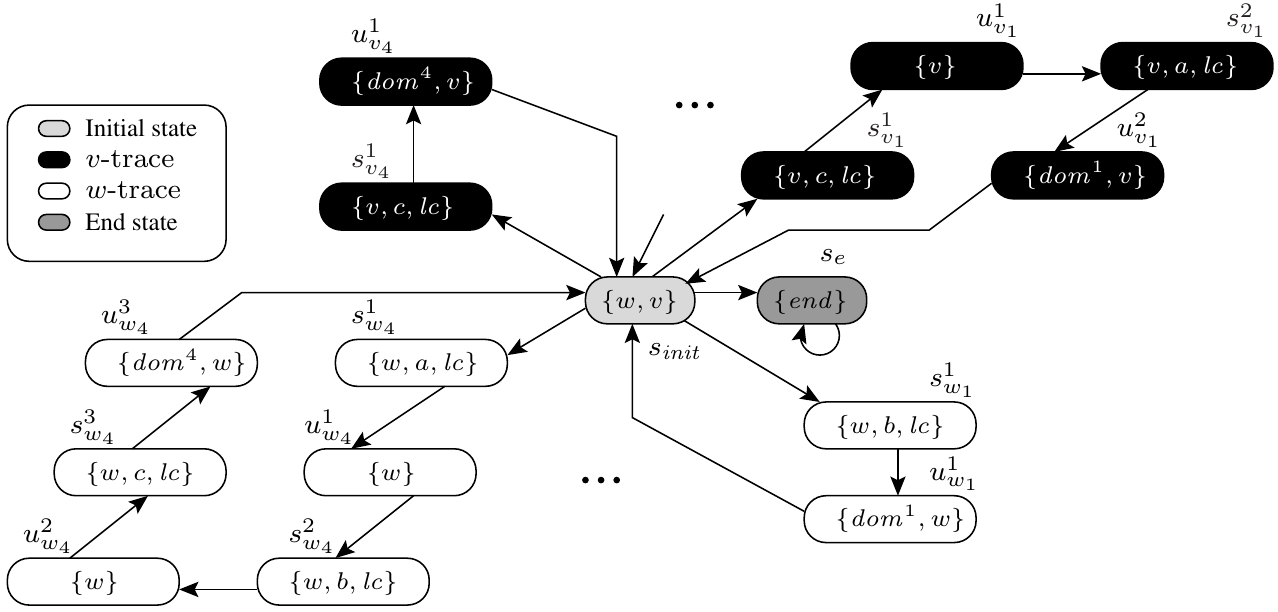}
	\caption{Mapping from PCP to model checking \AHLTL
		(only construction for dominos $\domi{w_1}{v_1} =\domi{b}{ca}$ and
		$\domi{w_4}{v_4} = \domi{abc}{c}$ are shown).}
	\label{fig:undec}
\end{figure}
PCP consists of a set of dominos, for example, of the form $\domi{w}{v} 
=\{\domi{b}{ca},\domi{a}{ab},\domi{ca}{a},\domi{abc}{c}\}$ and the problem is to decide whether there is 
a sequence of dominos (with possible repetitions), such that the upper and lower finite strings 
of the dominos are equal. A solution to the above set of dominos is the sequence       
$\domi{a}{ab}\domi{b}{ca}\domi{ca}{a}  \domi{a}{ab}\domi{abc}{c}$. We map a given set of dominos to 
a Kripke structure that allows arranging the dominos in a sequence (see Fig.~\ref{fig:undec} for an 
example), where $v$ and $w$ indicate lower and upper words, respectively, $\idx^i$ is for each domino 
$\domi{w_i}{v_i}$, and proposition $\lc$ marks whether or not a new letter is processed. The \AHLTL 
formula in our reduction is the following such that $\idx_{\pi_w} \DefinedAs 
\bigvee_{i\in[1..k]}\idx_{\pi_w}^i$:
 \begin{align*}
\nonumber	\varphiNoPCP \DefinedAs&\forall \pi_w\forall \pi_v. \Etau.\Big(\varphiType 
	\Then  
                                         ( \varphiSeqWV \vee \varphiLetter)\Big)   \\
   \text{where } \quad \quad
	\varphiType \DefinedAs & \Big((w_{\pi_w} \And \neg v_{\pi_w}) \U 
	\mathit{end}_{\pi_w}\Big) \And \Big((\neg w_{\pi_v} \And v_{\pi_v})  \U 
                                 \mathit{end}_{\pi_v} \Big) \notag \\
 \end{align*}   
 \begin{align*}
	\varphiSeqWV \DefinedAs & \G( \idx_{\pi_w}\Iff\idx_{\pi_v}) \And  \F\bigvee_{i=1}^k \idx_{\pi_w}^i 
                                  \not\Iff\idx^i_{\pi_v} \notag \\
	\varphiLetter \DefinedAs &  \G (\lc_{\pi_w} \Iff  \lc_{\pi_v}) \;
	\; \And \F\bigvee_{l \in \alphabet_{\mathit{pcp}}} 
                                   (l_{\pi_w}  \not\Iff  l_{\pi_v} \notag)
 \end{align*}
The intention of formula $\varphiNoPCP$ is that the Kripke structure
is a model of the formula if and only if the original PCP problem
has \emph{no} solution.
Intuitively, formula $\varphiType$ forces trace $\pi_w$ (respectively, $\pi_v$) 
to traverse only the traces labeled by $w$ (respectively, $v$) to build a 
$w$-word (respectively, $v$-word). Formula $\varphiSeqWV$ establishes that the trajectory aligns the
positions at which the domino indices are checked and at last once
the index is different.
Finally, formula $\varphiLetter$ captures if $\pi_w$ and $\pi_v$ are
aligned to compare the letters, at least one pair of the letters prescribed by the existential trajectory 
are different. In the detailed proof in the appendix, we show that the constructed Kripke structure 
satisfies formula $\varphiNoPCP$ if and only if the answer to deciding PCP is negative.\qed
\end{proof}

Theorem~\ref{thm:undec} above implies that there is no algorithm to
decide the model-checking problem correctly for every formula and
every system.
However, as we saw in Section~\ref{sec:decidable} for some formulas
the model-checking problem is decidable.
We now show that in these cases the problem is at least as hard as
model-checking HyperLTL, which is known to be
PSPACE-hard~\cite{clarkson14temporal,rabe16temporal}.

\newcounter{thm-lowerbound}
\setcounter{thm-lowerbound}{\value{theorem}}

\begin{theorem}
  \label{thm:lowerbound}
  Given a HyperLTL formula $\varphi$ and a Kripke structure $\krip$
  there is a \AHLTL formula $\varphi'$ and a Kripke structure $\krip'$
  such that $\krip'$ is linear in the size of $\krip$, $\varphi'$ is
  polynomial on the size of $\varphi$ and
  \( \krip\models\varphi \text{ if and only if } \krip'\models\varphi'\).
\end{theorem}
The proof proceeds as follow.  Giving $\krip$ we build a Kripke
structure $\krip'$ that alternates between real states in $\krip$ and
synchronization states.  Then the formula is transformed to force
alternations at every other step, therefore forcing the trajectory to
synchronize (see Appendix~\ref{subsec:pspace} for details).
Since the model-checking problem for HyperLTL is PSPACE-hard on the
size of the Kripke structure, the same follows for \AHLTL.

\newcounter{cor-pspace}
\setcounter{cor-pspace}{\value{corollary}}

\begin{corollary}
  \label{cor:pspace}
  For asynchronous HyperLTL formulas, the model checking problem is
  \comp{PSPACE-hard} in the size of the system.
\end{corollary}


%% file: casestudies.tex
\section{Case Studies and Evaluation}
\label{sec:eval}

We evaluated our algorithm for the $\forall^*_\pi \Etau $ \AHLTL fragment on several examples. 
After reducing the asynchronous model checking problem to a synchronous one, we use \textsc{MCHyper}~\cite{finkbeiner15algorithms, Hyperliveness} to check our property. 
\textsc{MCHyper} is a model checker for synchronous HyperLTL that can handle formulas with up to one quantifier alternation. 
It computes the self composition of the system and composes it with the formula automaton. 
\textsc{ABC}~\cite{conf/cav/BraytonM10} is then used as the backend tool checking the reachability of a violation. 

Our reduction from the asynchronous to the synchronous semantics follows the stuttering construction described in Section~\ref{subsec:stuttering}.
To model check a system against an \AHLTL formula, we first add a stuttering input to the system that forces the system to stutter in the current state.
The transformed formula ensures that the stuttering guarantees synchronous phase changes.
Currently, the transformation stage is only partially automated while
the synchronous verification is fully automated.
Future work includes the full automation of all steps.
We now describe the various case studies\footnote{All experimental
  data will be publicly available.}.
All our experiments were performed on a MacBook Pro with a $3.3$GHz
processor and $16$GB of RAM running MacOS~$11.1$. 

\subsection{Compiler Optimizations}

We modeled the source and target programs of different compiler
optimization techniques (from~\cite{WitnessingSecureCompilation}) as finite
state machines encoded as circuits, and used asynchronous
hyperproperties to prove the correspondence between both programs.
We analyzed the following optimizations:
\begin{compactitem}
\item Common Branch Factorization~(CBF), where expressions occurring
  in both branches of a conditional are factored out;
\item Loop Peeling~(LP), which consists in unrolling of a loop that is
  executed at least once;
\item Dead Branch Elimination~(DBE), that is, removing conditional
  checks and their branches that are unreachable; and
\item Expression Flatting~(EF), which splits complex computations into
  several explicit steps.
\end{compactitem}
Besides evaluating each optimization individually, we also examined several combinations of these optimizations.
Each optimization affects the alignment between source and target program, so synchronous hyperproperties fail to recognize the correspondence between both programs.
Using asynchronous hyperproperties instead allows us to compensate for this misalignment by stuttering the programs accordingly.
Essentially, each optimization is checked against the following \AHLTL formula in which $\pi$ represents traces from the source program and $\pi'$ traces from the target program: 
\[
	\forall \pi. \forall \pi'. \Etau. 
		(\bigwedge_{i \in I} i_\pi \leftrightarrow i_{\pi'}) 
	  \rightarrow 
		(\G \bigwedge_{o \in O} o_\pi \leftrightarrow o_{\pi'})
\]
This formula states that for all pairs of traces that initially agree on the inputs from the set~$I$ there exists a trajectory that aligns the phase changes of the outputs in set~$O$.
We use the stuttering construction and \textsc{MCHyper} to verify that
in all cases the source and target programs go through the same phases
of possibly different length.
The results of this case study are summarized in Table~\ref{table:CaseStudies}(a).
We note that \AHLTL model-checking subsumes the approach
in~\cite{WitnessingSecureCompilation} based on construction of a {\em
  buffer automaton} to reason about the alignment of executions.

\input{dualtable}

\subsection{SPI Bus Protocol}
The Serial Peripheral Interface (SPI) is a bus protocol that supports a single main component's communication with multiple secondary components. 
Each secondary can be selected individually by the main via the secondary's own $\mathit{ss}$ (``secondary select'') input signal.
If a secondary is enabled (that is, if $\neg \mathit{ss}$ holds as the secondary select is ``active low''), it reads the $\mathit{mosi}$ (main out, secondary in) signal and writes to the $\mathit{miso}$ (main in, secondary out) wire. 

We verify the behavior of a single SPI secondary component
that receives an input which it sends to the main component upon
request.
This behavior should always be the same, independent of when the
secondary is enabled or how fast the bus protocol's ``serial clock''
($\mathit{sclk}$) set by the main component ticks compared to the secondary's internal clock.
The \AHLTL formula we check is the following (see observational determinism in Section~\ref{form:obsdet}):

\newcommand{\SSe}{\ensuremath{\mathit{ss}}}
\newcommand{\SCLK}{\ensuremath{\mathit{sclk}}}
\newcommand{\MISO}{\ensuremath{\mathit{miso}}}
\newcommand{\SSpi}{\ensuremath{\SSe_\pi}}
\newcommand{\SSpiP}{\ensuremath{\SSe_{\pi'}}}
\newcommand{\IN}{\ensuremath{\mathit{in}}}
\newcommand{\INIT}{\ensuremath{\mathit{init}}}

\iftrue
\begin{equation*}
  \forall \pi. \forall \pi'. \Etau. 
  \begin{pmatrix}
    \begin{array}{c}
      \bigwedge\limits_{i \in \{in, \mathit{init}\}} i_\pi \Iff i_{\pi'} \\ \And \\
      \textit{SPI input assumptions}
    \end{array}
    \end{pmatrix} 
  \Then
  \Always
  \begin{pmatrix}
    \begin{array}{c}
      (\MISO_\pi \And \neg \SCLK_\pi \And \neg \SSpi) \\ \Iff \\
      (\MISO_{\pi'} \And \neg \SCLK_{\pi'} \And \neg \SSpiP)
    \end{array}
    \end{pmatrix}
  \label{form:SPI}  
\end{equation*}
\fi

%

\noindent This formula (called \SPICorrect in Table~\ref{table:CaseStudies}(b))
ensures that for all pairs of traces~$\pi$ and $\pi'$ that agree on
the initial configuration, on the input, and additional \textit{SPI input assumptions}, there is a trajectory that
aligns their relevant behavior.
We consider it relevant that both secondaries agree on their~$\mathit{miso}$ output whenever they are enabled and the~$\mathit{sclk}$ is low.
Checking~$\mathit{miso}$ only when the~$\mathit{sclk}$ is low is sufficient as changes on~$\mathit{miso}$ only occur at falling edges of the~$\mathit{sclk}$. 
The \textit{SPI input assumptions} are required to guarantee the implicit assumptions of the protocol, for example, that the~$\mathit{sclk}$ behaves as an infinitely ticking clock.  
By introducing additional variables and applying logical transformations, we obtain an equivalent formula that syntactically lies in the fragment of the stuttering construction. 
Again, we reduce this model checking problem to the synchronous
semantics and use \textsc{MCHyper} to perform the verification.

In a second experiment, we modified the system to send the value only once and checked it for termination insensitive noninterference \SPITerm (see Sections~\ref{form:TNI1} and~\ref{form:TNI2}).
In our setup, we use the variable $\mathit{term}$ to flag that the secondary has sent the full value. 
In the premise of the formula, we require that the input value is equal on both traces and again assume that the inputs conform to the SPI protocol. 
The conclusion checks if both secondaries have sent the same values by using additional variables that are set together with $\mathit{term}$. 
The results of this case study are summarized in
Table~\ref{table:CaseStudies}(b).
%


%% file: dualtable.tex
\begin{table}[t]
  \centering
  \begin{tabular}{cc}
  \begin{tabular}{l c S[table-format=4.1] S[table-format=4.1]}
    Optimizations \hspace{0.5em} & \multicolumn{2}{c}{System Size \hspace{0.5em}} & {Time (s) \hspace{0.5em}} \\ 
                               & \#\textsc{ls} \hspace{1em} & {\#\textsc{ands}} & \\ \hline \hline
    EF             & 12 & 64   & 0.6    \\ 
    DBE            & 16 & 128  & 0.8    \\ 
    CBF            & 16 & 145  & 2.7    \\ 
    LP             & 28 & 514  & 365.9  \\ \hline 
    CBF+DBE      & 16 & 137  & 11.4   \\ 
    CBF+DBE+EF & 20 & 175  & 10.0   \\ 
    CBF+EF       & 20 & 180  & 1.7    \\ 
    EF+LP        & 41 & 8642 & 1315.2    
  \end{tabular} & 
    \begin{tabular}{l c S[table-format=4.1] S[table-format=4.1]}
    Propery \hspace{0.5em} & \multicolumn{2}{c}{System Size \hspace{0.5em}} & {Time (s) \hspace{0.5em}} \\ 
                               & \#\textsc{ls} \hspace{1em} & {\#\textsc{ands}} & \\ \hline \hline
    \SPICorrect  & 30 & 175 & 65.7 \\ 
    \SPITerm & 33 & 296 & 155.8 
  \end{tabular}
  \\
  (a) Compiler Optimizations & (b) SPI \\[0.3em]
  \end{tabular}
  \caption{Verification times and system sizes in number of latches (\#\textsc{ls})
    and \textsc{and}-gates (\#\textsc{ands}) for the case studies.}
  \label{table:CaseStudies}
\end{table}


%% file: related.tex
\section{Related Work}
\label{sec:related}

The study of specific hyperproperties, such as noninterference, dates
back to the seminal work by Goguen and Meseguer \cite{gm82} in the
1980s.
The first systematic study of hyperproperties is due to Clarkson and
Schneider~\cite{cs10}.

It is well-known that classic specification languages like LTL cannot
express hyperproperties.
There are two principal methods with which the standard logics have
been extended to express hyperproperties:
\begin{itemize}
\item The first method is the quantification over variables that
  identify specific paths or traces.
  The temporal logics LTL, CTL$^*$ have been extended with
  quantification over traces and paths, resulting in the temporal
  logics HyperLTL and HyperCTL$^*$~\cite{clarkson14temporal}.
  There are also extensions of the $\mu$-calculus, most recently, the
  temporal fixpoint calculus
  $H_\mu$~\cite{DBLP:journals/pacmpl/GutsfeldMO21}, which adds path
  quantifiers to the polyadic
  $\mu$-calculus~\cite{Andersen/1994/APolyadicModalMuCalculus}.
\item The second method is the addition of the equal-level predicate
  $E$ to first-order and second-order logics, like MPL, MSO, FOL, and
  S1S, which results in the logics FOL[E], S1S[E], MPL[E],
  MSO[E]~\cite{Hierarchy,Martin}.
\end{itemize}

HyperCTL$^*$, MPL[E], and MSO[E] are branching-time logics, we
therefore focus in the following on the linear-time logics HyperLTL,
$H_\mu$, FOL[E], and S1S[E].  Among these logics, HyperLTL is the only
logic for which practical model-checking algorithms are
known~\cite{finkbeiner15algorithms,Hyperliveness,hsb21}.
For HyperLTL, the algorithms have been implemented in the model
checkers MCHyper and bounded model checker HyperQube.
As discussed in this paper, HyperLTL is limited to synchronous
hyperproperties.

FOL[E] can express a limited form of asynchronous hyperproperties.
As shown in~\cite{Hierarchy}, FOL[E] is subsumed by HyperLTL with
additional quantification over predicates.
Using such predicates as ``markers,'' one can relate different
positions in different traces.
However, only a finite number of such predicates is available in each
formula.
S1S[E] is known to be strictly more expressive than
FOL[E]~\cite{Hierarchy}, and conjectured to subsume
$H_\mu$~\cite{DBLP:journals/pacmpl/GutsfeldMO21}.
For $S1S[E]$ and $H_\mu$, the model checking problem is in general
undecidable; for $H_\mu$, two fragments, the $k$-synchronous,
$k$-context bounded fragments, have been identified for which model
checking remains decidable~\cite{DBLP:journals/pacmpl/GutsfeldMO21}.
It is not known, however, if any of the commonly used hyperproperties,
like observational determinism, noninterference, or linearizability,
can be encoded in these fragments.

Like S1S[E] and $H_\mu$, asynchronous HyperLTL has an (in general)
undecidable model checking problem.
However, in this paper we have identified decidable fragments of
asynchronous HyperLTL that can express observational determinism,
noninterference, and linearizability.
\AHLTL is thus the first logic for hyperproperties that can express
the major asynchronous hyperproperties of interest within decidable
fragments.
Furthermore, asynchronous HyperLTL is the first logic for asynchronous
hyperproperties with a practical model checking algorithm.




%% file: conclusion.tex
\section{Conclusion}
\label{sec:conclusion}

We have introduced \AHLTL, a temporal logic to describe asynchronous
hyperproperties.
This logic extends HyperLTL with {\em trajectory} modalities, which control
when a trace proceeds and when it stutters.
Synchronous HyperLTL corresponds to a trajectory that always moves all
paths in a lock-step manner.
%
%
This notion of trajectory allows to define formulas that are invariant
under stuttering, paving the way for relevant model-checking
optimizations such a partial order reduction and
abstraction-refinement techniques in the context of hyperproperties.
We show that model-checking \AHLTL formulas is in general undecidable,
and identify two fragments of \AHLTL formulas, which cover a rich set
of security requirements and can be decided by a reduction to
HyperLTL model-checking.
This in turn has allowed us to the reuse the existing model-checker
McHyper.

Future work includes the study of larger decidable fragments (that
encompass both fragments studied here), extending the logic allowing
several trajectory modalities, as well as their implementation in
practical tools.
Extending bounded model-checking~\cite{hsb21} to \AHLTL is another
interesting research problem.
Asynchronous hyperproperties are important for applying a logic-based
verification approach to verify hyperproperties for {\em software}
programs, because the relative speed of the execution of programs
depends on many factors like the compiler, hardware, execution
platform and concurrent running programs, that the analysis must
tolerate.
Therefore, future work includes adapting techniques for infinite-state
software model-checking, like deductive methods, abstraction, etc to
verify \AHLTL properties of software systems.

%
%


%% file: proofs.tex
\section{Full Proofs of Undecidability and Complexity}
\label{sec:app:undec-pspace}

\subsection{Undecidability of Asynchronous HyperLTL}
\label{subsec:undec}

We show now the full proof of undecidability of full \AHLTL via
reduction from PCP (Theorem~\ref{thm:undec})

\newcounter{aux}
\setcounter{aux}{\value{theorem}}
\setcounter{theorem}{\value{thm-undec}}

\begin{theorem}
  Let $\krip$ be a Kripke structure and $\varphi$ be an asynchronous
  HyperLTL formula. The problem of determining whether or not
  $\krip \models \varphi$ is undecidable.
\end{theorem}

\setcounter{theorem}{\value{aux}}

\begin{proof}	
  In order to show the undecidability of model-checking A-HLTL, we
  reduce the complement of the {\em post correspondence problem
    (PCP)}, which is a well-known undecidable
  problem~\cite{post47recursive,sipser12introduction}, to our model
  checking problem.
  The PCP is as follows.
  Let $\alphabet_{\mathit{pcp}}$ be an alphabet of size at least two.
  The input to the decision problem is a collection of dominos of the
  form:
  \[
    \Domi{w}{v} =
    \bigg\{\Domi{w_1}{v_1}, \Domi{w_2}{v_2},\dots, \Domi{w_k}{v_k} \bigg\} \]
  where for all $i \in [1, k]$, we have $v_i,w_i \in \alphabet_{\mathit{pcp}}^*$.
  The decision problem consists of checking whether there exists a
  finite sequence of dominos of the form
  \[
    \Domi{w_{i_1}}{v_{i_1}}\Domi{w_{i_2}}{v_{i_2}} \cdots \Domi{w_{i_m}}{v_{i_m}}
  \]
  where each index $i_j \in [1, k]$, such that the upper and lower finite strings 
  of the dominos are equal, i.e.,
  
  \[ v_{i_1}v_{i_2}\cdots{}v_{i_m} = w_{i_1}w_{i_2}\cdots{}w_{i_m}. \]
  For some collection of dominos, it is impossible to find a match. An example is 
  the following:
  \[
    \Domi{w}{v} = \bigg\{\Domi{abc}{ab}, \Domi{ca}{a}, \Domi{acc}{ba}\bigg\}
  \]

  We now start from an arbitrary instance 
  $\domi{w}{v}=\{\domi{w_2}{v_2}, \domi{w_2}{v_2}, \dots,
  \domi{w_k}{v_k}\}$ of the PCP
  and create a Kripke structure $\krip_{\mathit{pcp}} = \ktuple$ and
  an asynchronous HyperLTL formula $\varphiNoPCP$ (i.e., no solution
  to PCP), and later show that the PCP instance has {\em no} solution
  if and only if $\krip_{\mathit{pcp}} \models \varphiNoPCP$.
  To illustrate the construction, we use the following instance of PCP
  as a running example.
  The alphabet is $\alphabet_{\mathit{pcp}} = \{a, b, c\}$ and the
  set of dominos~\cite{sipser12introduction} is:
  \[
    \Domi{w}{v} =\bigg\{\Domi{b}{ca},\Domi{a}{ab},\Domi{ca}{a},\Domi{abc}{c} \bigg\}
  \]
  
  \paragraph{Kripke structure.} We map the PCP instance to 
  $\krip_{\mathit{pcp}}$ as follows (see Fig.~\ref{fig:undec} for 
  example):

  \begin{itemize}
  \item \emph{Atomic propositions}: The set $\AP$ of atomic propositions
    include the following:
    \begin{itemize}
    \item Two propositions $v$ and $w$ to indicate that a state belongs to
      a path that builds a $v$-word in the instance of PCP, or to a path
      that builds a $w$-word.
    \item A proposition $\idx^i$ for each domino from $\domi{w_i}{v_i}$,
      where $i \in [1, k]$, that indicates the domino that contains a word
      from $v_i$ or $w_i$ that has been processed.
    \item A {\em letter clock} proposition $\lc$, which serves the purpose
      of marking whether or not a new letter is processed.
    \item A proposition $\mathit{end}$, which serves the purpose of
      marking the end of sequencing dominos.
    \item A proposition for each letter in the PCP alphabet
      $\alphabet_{\mathit{pcp}}$.
    \end{itemize}
    Thus, we have:
    \[
      \AP = \alphabet_{\mathit{pcp}} \, \cup \, \{v, w, \lc, \mathit{end}\} \, \cup 
      \,\{\idx^i \mid i \in [1,k]\}
    \]
  \item \emph{States and transitions}:
    We include the following states $S$ and transitions $\trans$:
    \begin{itemize}
    \item We include a single initial state $s_\init$ with no labels and
      an end state $s_e$ labeled by $\{\mathit{end}\}$ with a self-loop
      $(s_e, s_e)$.
    \item For each word $w_i = w_i^1\cdots w_i^l$ in a domino
      $\domi{w_i}{v_i}$ of the PCP instance, where $i \in [1,k]$, we
      include a path of length $2l$ that starts and ends in $s_\init$,
      and whose states are only reachable in the given path.
      More specifically,
      \begin{itemize}
      \item For each $j \in [1,l]$, we introduce two states $s_{w_i}^j$
        and $u_{w_i}^j$.
        Each state $s_{w_i}^j$, where $i \in [1,k]$, is labeled by
        $\{w, w_i^j, \lc\}$.
        This state indicates a new letter in the word $w_i$.
        Each state $u_{w_i}^j$, where $j \in [1,l)$, is labeled by
        $\{w\}$.
        This state does not correspond to a letter and its purpose is
        to server as separation between two letters within the same
        word.
        In this manner, a word is encoded by a sequence of
        letter/no-letter states.
        Finally, state $u_{w_i}^l$ is labeled by $\{\idx^i, w\}$,
        which is used for synchronizing the end of words $w_i$ and
        $v_i$ in domino $i$.
      \item We add a transition $(s^j_{w_i}, u^j_{w_i})$ for each
        $j \in [1, l]$ and a transition $(u^j_{w_i}, s^{j+1}_{w_i})$ for
        each $j \in [1, l)$.
      \item We also add transitions $(s_\init, s_{w_i}^1)$,
        $(u_{w_i}^l, s_\init)$.
      \end{itemize}
      
    \item For each word $v_i = v_i^1\cdots v_i^l$ in a domino
      $\domi{w_i}{v_i}$ of the PCP instance, where $i \in [1,k]$, we
      include a path in the same way as we did for $w_i$, except that
      all instances of ``$w$'' are replaced by ``$v$''.
    \end{itemize}
  \end{itemize}

  \paragraph{Asynchronous HyperLTL formula}
  The formula that will be checked against the Kripke structure is the
  following:
  
  \begin{align}
    \varphiNoPCP \DefinedAs&\forall \pi_w\forall \pi_v. \Etau.\Big(\varphiType 
                             \Then  
                             ( \varphiSeqWV \vee \varphiLetter)\Big)     \label{eq:ahltl} \\
    \intertext{where}
    \varphiType \DefinedAs & \Big((w_{\pi_w} \And \neg v_{\pi_w}) \U 
                             \mathit{end}_{\pi_w}\Big) \And \Big((\neg w_{\pi_v} \And v_{\pi_v})  \U 
                             \mathit{end}_{\pi_v} \Big) \notag \\
    \varphiSeqWV \DefinedAs & \G( \idx_{\pi_w}\Iff\idx_{\pi_v}) \And  \F\bigvee_{i=1}^k \idx_{\pi_w}^i 
                              \not\Iff\idx^i_{\pi_v} \notag \\
    \varphiLetter \DefinedAs &  \G (\lc_{\pi_w} \Iff  \lc_{\pi_v}) \;
                               \; \And \F\bigvee_{l \in \alphabet_{\mathit{pcp}}} 
                               (l_{\pi_w}  \not\Iff  l_{\pi_v} \notag)\\
    \intertext{and}
    \nonumber \idx_{\pi_w} \DefinedAs & \bigvee_{i\in[1..k]}\idx_{\pi_w}^i
  \end{align}

  The intention of formula $\varphiNoPCP$ is that the Kripke structure
  is a model of the formula if and only if the original PCP problem
  has \emph{no} solution.
  Intuitively, formula $\varphiType$ forces trace $\pi_w$ (respectively, $\pi_v$) 
  to traverse only the traces labeled by $w$ (respectively, $v$) to build a 
  $w$-word (respectively, $v$-word). When such traces are built (i.e., the left side of the implication), 
  PCP has no solution when at least one of the following conditions hold: either (1) the order of 
  dominos in $\pi_w$ and $\pi_v$ is not legitimate, or (2) the letters do not match. 
  More formally, formula $\varphiSeqWV$ establishes that the trajectory aligns the
  positions at which the domino indices are checked and at last once
  the index is different.
  Finally, formula $\varphiLetter$ captures if $\pi_w$ and $\pi_v$ are
  aligned to compare the letters, at least one pair of the letters prescribed by the existential trajectory 
  are different.
  
  We now show that the given PCP instance has no matching if and only
  if $\kripPCP \not\models \varphiNoPCP$:

  \begin{itemize}
  \item ($\Rightarrow$) Supposed that the PCP instance has no
    solution.  That is, no sequence of the dominos of the form:
    \[
      \Domi{w_{i_1}}{v_{i_1}}\Domi{w_{i_2}}{v_{i_2}} \cdots 
      \Domi{w_{i_m}}{v_{i_m}}
    \]
    constructs a matching where $w_{i_1}w_{i_1}\ldots{}w_{i_m}=v_{i_1}v_{i_1}\ldots{}v_{i_m}$.
    In our running example, the following is a sequence of dominos that does not render a good 
    matching:
    \begin{align}
      \label{eq:ex1}
      \Domi{a}{ab}\Domi{b}{ca}\Domi{ca}{a} \Domi{abc}{c}
    \end{align}
    But the following is a solution to the problem:
    \begin{align}
      \label{eq:ex2}
      \Domi{a}{ab}\Domi{b}{ca}\Domi{ca}{a}  \Domi{a}{ab}\Domi{abc}{c}
    \end{align}
    We now show that if the PCP problem has no solution, then $\kripPCP 
    \models \varphiNoPCP$. 
    Let $\pi_w$ and $\pi_v$ be any two arbitrary traces of the $\kripPCP$. If these traces do not 
    satisfy 
    formula $\varphiType$, then they are representing two words that are actually not viable $w$ and 
    $v$ words. This makes $\varphiNoPCP$ vacuously true. So, let us consider the case where 
    $\pi_w$ 
    and $\pi_v$ satisfy $\varphiType$. For example, Fig.~\ref{fig:pcptraj1} shows an instantiation of 
    traces $\pi_w$ and $\pi_v$ with sequence~\eqref{eq:ex1} that satisfies 
    $\varphiType$.
    
    Now, we consider two cases: either the words represented by $\pi_w$ and 
    $\pi_v$ do not establish a legitimate sequence of dominos or they do. 
    If they do not, then traces $\pi_w$ and $\pi_v$ satisfy formula $\varphiSeqWV$ 
    and, hence, we have $\kripPCP \models \varphiNoPCP$. Now, consider the case 
    where $\pi_w$ and $\pi_v$ do establish a legitimate sequence of dominos. Since 
    there is no solution to PCP, the words represented by $\pi_w$ and $\pi_v$ 
    cannot match. That means there is at least one pair of letters that do not 
    match, which in turn implies the existence of a trajectory that finds this 
    mistmatch and indeed satifies formula $\varphiLetter$. Put it another 
    way, sine there is no solution the to PCP, the sequence that reaches a letter 
    mismatch construct the witness to the existential trajectory quantifier.  
    Fig.~\ref{fig:pcptraj1} shows the trajectory that identifies the mismatched 
    letters in our running example. Therefore, we have $\kripPCP \models 
    \varphiNoPCP$.
    
    \begin{figure}[t]
      \centering
      \includegraphics[width=\textwidth]{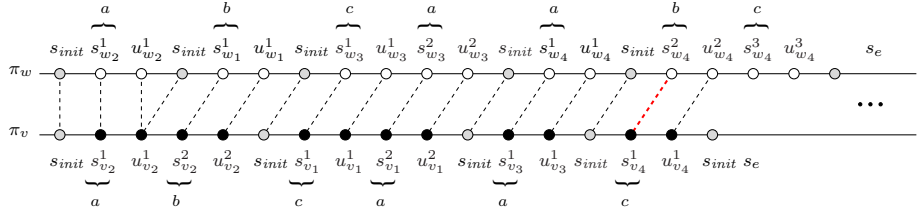}
      \caption{Obtaining PCP matching from an instantiation of trajectory 
        $\tjvar_2$ in Formula~\ref{eq:ahltl}.}
      \label{fig:pcptraj1}
    \end{figure}

  \item ($\Rightarrow$) Now, suppose we have $\kripPCP \models \varphiNoPCP$. 
    This 
    implies that for any pair of traces $\pi_v$ and $\pi_w$ that represent legitimate $w$ and $v$ 
    words 
    (i.e., where $\varphiType$ holds), there exists a trajectory that satisfies $(\varphiSeqWV \vee 
    \varphiLetter)$. This means that for the chosen pair of traces $\pi_v$ and $\pi_w$, the existential 
    trajectory has found either illegitimate alignment of dominos, or letters that do not match. 
    In our example, the existential trajectory can be instantiated by the 
    trajectory shown in Fig.~\ref{fig:pcptraj1}. Such a witness trajectory shows 
    that for a pair of traces that represent ligitimate $w$ and $v$ words, 
    either they do not establish a legitimate sequence of dominos or not all 
    letters along the words match.
    Since this holds for any such pair of traces that explore any ordering of $w$ and $v$ words, it 
    means that there is no solution to PCP.
  \end{itemize}
  This shows that the $\kripPCP \models \varphiNoPCP$ if and only if the PCP 
  instance $\frac{w}{v}$ has no solution. Consequently, the model-checking problem 
  for asynchronous HyperLTL is undecidable.
  \qed
\end{proof}

\subsection{Lower Bound Complexity}
\label{subsec:pspace}

We show now construction of the lower bound complexity using a
reduction to HyperLTL model-checking.
The main results is Corollary~\ref{cor:pspace} that follows from the
following theorem.

\setcounter{aux}{\value{theorem}}
\setcounter{theorem}{\value{thm-lowerbound}}

\begin{theorem}
  Given a HyperLTL formula $\varphi$ and a Kripke structure $\krip$
  there is a \AHLTL formula $\varphi'$ and a Kripke structure $\krip'$
  such that $\krip'$ is linear in the size of $\krip$, $\varphi'$ is
  polynomial on the size of $\varphi$ and
  \[ \krip\models\varphi \text{ if and only if } \krip'\models\varphi'\]
\end{theorem}

\setcounter{theorem}{\value{aux}}


\begin{proof}
We reduce the HyperLTL model-checking problem, which is known to be
\comp{PSPACE-hard} to asynchronous HyperLTL model-checking.

Consider an instance of the HyperLTL model-checking problem consist
of a Kripke structure $\krip = \ktuple$ and a HyperLTL formula of
the form:
\[
  \varphi = Q_1 \trvar_1.Q_2 \trvar_2\cdots Q_n  \trvar_n.\psi
\]
where each $Q_i \in \{\forall, \exists\}$ is a trace quantifier
($i \in [1,n]$).
We map this instance to asynchronous HyperLTL model-checking problem
as follows:
\begin{itemize}
\item {\bf Kripke structure. }
  We now transform $\krip$ into another Kripke structure
  $\krip' = \langle S', s_\init, \trans', L'\rangle$ as follows:
  \begin{itemize}
  \item First, we set the set of atomic propositions as
    $\AP \cup \{\sync\}$.
  \item For each transition $(s_0, s_1) \in \trans$, we introduce a new state 
    $u_{(s_0, s_1)}$ labeled only by proposition $\sync$.
    That is:
    \[ S' = S \cup \{u_{(s_0, s_1)} \mid (s_0, s_1) \in \delta\} \]
    The set of labels of states in $S$ remains unchanged.
  \item We replace transition $(s_0, s_1)$ by two transitions $(s_0, u_{(s_0, 
      s_1)})$ and $(u_{(s_0, s_1)}, s_1)$:

    \vspace{0.5em}
    \framebox{
      \includegraphics[scale=1]{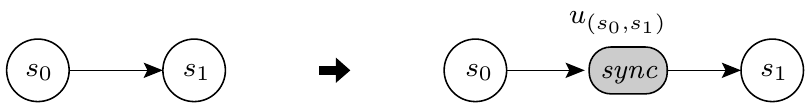}
    }
    \vspace{0.5em}
  
    That is,
    \(
      \delta' = \{(s_0, u_{(s_0, s_1)}), (u_{(s_0, s_1)}, s_1) \mid  (s_0, s_1) \in 
      \delta\}. \)
  \end{itemize}
  
\item {\bf Formula. } We first map subformulas of $\psi$ using the
  following $g$ function:
  \begin{align*}
    g(a_\pi) & = a_\pi\\
    g(\neg \psi) & = \neg g(\psi)\\
    g(\psi_1 \vee \psi_2) & = g(\psi_1) \vee g(\psi_2)\\
    g(\Next \psi) & = \Next \Next g(\psi)\\
    g(\psi_1 \U \psi_2) & = \bigg(\bigwedge_{i=1}^n \neg \sync_{\pi_i} \Into 
                          g(\psi_1)\bigg) \U \bigg(g(\psi_2) \, \wedge \, \bigwedge_{i = 1}^n \neg 
                          \sync_{\trvar_i}\bigg)
  \end{align*}
  Intuitively, function $g$ ensures that (1) state formulas are
  evaluated in the same fashion as in HyperLTL semantics (2) the
  next operator is evaluated within two steps (instead of one) due
  to the existence of intermediate $\sync$ states, and (3) both
  sides of the until operator are only evaluated over $\sync$
  states.
  The asynchronous HyperLTL formula will be the following:
  \[
    \varphi' \DefinedAs Q_1 \trvar_1.Q_2 \trvar_2\cdots Q_n 
    \trvar_n.\Etau.\bigg(g(\varphi) \And \psiSync\bigg)
  \]
  \noindent where
  \[
    \psiSync \DefinedAs  \G_\tjvar \bigg(\bigwedge_{i,j = 1}^n 
    (\sync_{\trvar_i} \Iff \sync_{\trvar_j}) \And
    (\bigwedge_{i = 1}^n \sync_{\trvar_i} \Iff \Next 
    \bigwedge_{i = 1}^n  \neg \sync_{\trvar_i})
  \]
\end{itemize}
Intuitively, formula $\psiSync$ enforces the synchronous semantics.
That is, each step of trajectory $\tjvar$ alternates between all
$\sync$ and all $\neg \sync$ states across all traces.
Formula $\psiMap$ ensures that all subformulas of the given HyperLTL
formula are evaluated only in $\sync$ states across all traces.

We now show that $\krip \models \varphi$ if and only if $\krip' \models 
\varphi'$:
\begin{itemize}
\item ($\Rightarrow$) Suppose that $\krip \models \varphi$ holds.
  To show that $\krip' \models \varphi'$, we instantiate all
  existential trace quantifiers in $\varphi'$ as we do for $\krip$
  to satisfy $\varphi$ (except that traces in
  $\Traces(\krip', s_\init)$ have the additional intermediate
  $\sync$ state in between each two consecutive states of traces of
  $\krip$).
  Moreover, we instantiate $\tjvar$ in $\varphi'$ with a trajectory
  that advances all traces in all its steps (i.e., identical to the
  semantics of synchronous HyperLTL).
  It is straightforward to see that this trajectory satisfies
  $\psiSync$, as in each step of the trajectory, $\sync$ is
  alternatively true and false either across all traces or none.
  We now show that $\psiMap$ is also satisfied by structural
  induction:
  \begin{itemize}
  \item If $\psiMap$ is a state formula (i.e., an atomic
    proposition, negation or disjunction), then it is trivially
    satisfied, as $g(\psi) = \psi$.
  \item If $\Next\psi$ holds in $\krip$, then $\psi$ holds in two
    steps in $\tjvar$, as an intermediate $\sync$ state is inserted
    in each transition of $\trans$ of $\krip'$.
    This implies that $\Next_\tjvar\Next_\tjvar \psi$ holds in
    $\krip'$.
  \item Suppose $\psi_1 \U \psi_2$ holds in $\krip$. Then, since
    there exists some step $i \geq 0$ in instantiated traces from
    $\Traces(\krip, s_\init)$ such that $\psi_2$ holds.
    This means that in the instantiated traces from
    $\Traces(\krip', s_\init)$, $g(\psi_2)$ holds in $2i$ steps,
    where $\sync$ is false across all traces.
    This the right side of the until holds in $\krip'$.
    Also, for all $j \in [0, i)$, in instantiated traces from
    $\Traces(\krip, s_\init)$, $\psi_1$ holds.
    This means that in instantiated traces from
    $\Traces(\krip', s_\init)$, $g(\psi_1)$ also holds if $\sync$ is
    false in states across all traces.
  \end{itemize}
  Hence, we have $\krip' \models \varphi'$.
\item ($\Leftarrow$) Suppose that we have $\krip' \models
  \varphi'$.
  To show that $\krip \models \varphi$, we instantiate all
  existential trace quantifiers in $\varphi$ as we do for $\krip'$
  to satisfy $\varphi'$ (except that in traces of
  $\Traces(\krip, s_\init)$, intermediate $u$ states in between each
  two consecutive states of traces of $\krip'$ are
  eliminated).
  Since, $\psiSync$ holds in $\varphi'$, there exists a trajectory
  $\tjvar$ which advances all traces in all its steps such that in
  each step of the trajectory, $\sync$ is alternatively true and
  false either across all traces or none.
  We now show that $\varphi$ holds in $\krip$ by structural
  induction:
  \begin{itemize}
  \item If $\psi$ is a state formula (i.e., an atomic proposition,
    negation or disjunction), then it is trivially satisfied, as
    $\psi = g(\psi)$.
  \item First, notice that the number of nested next operators in
    $\varphi'$ is always even.
    If $\Next\Next\psi$ holds in $\krip'$, then $\psi$
    holds in two steps in $\tjvar$.
    When the $u$ state is eliminated, $\psi$ will hold in the next
    state. This implies that $\Next \psi$ holds in $\krip$.
  \item Suppose $\psi_1 \U \psi_2$ holds in $\krip'$.
    Then, there exists some step $i \geq 0$ in instantiated traces
    from $\Traces(\krip', s_\init)$ such that $\psi_2$ holds.
    In this step, $\sync$ has to be false states across all
    traces.
    This means that in the instantiated traces from
    $\Traces(\krip, s_\init)$, $g^{-1}(\psi_2)$ (obtained by
    removing the right side of the conjunction when computing $g$ of
    an until) holds in $i/2$ steps.
    Also, $\psi_1$ holds in all states $j \in [0, i)$, in
    instantiated traces from $\Traces(\krip', s_\init)$.
    This means that in all $j \in [0, i/2)$ in the instantiated
    traces from $\Traces(\krip, s_\init)$, $g^{-1}(\psi_1)$
    (obtained by removing the left side of the implication when
    computing $g$) holds.
  \end{itemize}
  It is easy to see that the size of $\krip'$ is linear in the size
  of $\krip$ and that $\varphi'$ is polynomial on the size of
  $\varphi$.
  \qed
\end{itemize}
\end{proof}

\setcounter{aux}{\value{corollary}}
\setcounter{corollary}{\value{cor-pspace}}

\begin{corollary}
  For asynchronous HyperLTL formulas, the model-checking problem is
  \comp{PSPACE-hard} in the size of the system.
\end{corollary}

\setcounter{corollary}{\value{aux}}

        

%% file: decproofs.tex
\vfill\pagebreak
\section{Full Decidability Proofs}
\label{app:decidability}

\subsection{The Stuttering Construction}

Since there is only one occurrence of $\psiPH$ in $\psi$ we present
here the case where the occurrence is with positive polarity as in
Section~\ref{sec:decidable}.
The case for negative polarity is analogous.
We will use the fact that $\psiPH$ occurs with positive polarity to
spuriously evaluate $\psi'$ (the formula that $\psiPH$ gets translated
into) into $\true$ for those tuples of paths of $\kripST$ for which the
evaluation is not relevant (non-aligned traces).
The final verdict will depend on those paths that are aligned.

For the following discussion we fix $\krip$ as the original Kripke
structure, $\kripST$ the corresponding stuttering structure and $\psi$
as the admissible formula.
Given a tuple of traces $(\sigma_1,\ldots,\sigma_n)$ of $\krip$ and a
trajectory $t$, we define the following sequence of tuples of pointed
traces
\[ 
  \big((\sigma_1,p^0_1),\ldots,(\sigma_n,p^0_n)\big)\;\;\;\;,\;\;\;\;
  \big((\sigma_1,p^1_1),\ldots,(\sigma_n,p^1_n)\big) \;\;\;\cdots
\]
such that (1) $p_i^0=0$ and (2) $p^{k+1}_i=p^k_i+1$ if
$\pi_i\in t(k)$ and $p^{k+1}_i=p^k_i$ if $\pi_i\notin t(k)$.
In turn, this sequence induces a tuple of traces
$(\sigma_1^*,\ldots,\sigma^*_n)$ of $\kripST$ as follows:
\[
  \begin{array}{rcl}
    \sigma_i^*(0) &=& \sigma_i(0)\\
    \sigma_i^*(k+1)&=&\begin{cases}
      \sigma_i(p_i^{k+1}) & \text{if } p_i^{k+1}=p_i^k \\
      \sigma_i(p_i^{k+1})\cup\{\ST\} & \text{if } p_i^{k+1}=p_i^k+1 \\
    \end{cases}
  \end{array}
\]
Note that $\sigma^*_i$ is encoding the progress of $\sigma_i$
according to trajectory $t$.
It is easy to see that $(\sigma^*_1,\ldots,\sigma_n^*)$ are fair
traces of $\kripST$ (due to the fairness requirement of $t$) and that
$\sigma_i^*$ is a stuttering expansion of $\sigma_i$.
We call $(\sigma^*_1,\ldots,\sigma_n^*)$ the traces induced by
$(\sigma_1,\ldots,\sigma_n)$ and $t$, which we denote by
\[
  (\sigma_1^*,\ldots,\sigma_n^*)=\Expand((\sigma_1,\ldots,\sigma_n),t).
\]

Similarly, given a tuple of fair traces
$(\sigma^*_1,\ldots,\sigma_n^*)$ of $\kripST$ there is a unique tuple
$(\sigma_1,\ldots,\sigma_n)$ of traces of $\krip$ and trajectory $t$
defined as follows.
Given $\sigma_i^*$, let $i_1<i_2<\ldots$ be the ordered set of
positions for which $\ST\notin\sigma_i^*$ (this set is guaranteed to
be infinite because $\sigma_i^*$ is fair).
Then $\sigma_i$ is defined as $\sigma_i(k)=\sigma_i^*(i_k)$.
Also, $t(k)$ contains $\pi_i$ if and only if $\ST\notin\sigma^*_i(k)$.
We call $(\sigma_1,\ldots,\sigma_n)$ and $t$ the trace and trajectory
induced by $(\sigma_1^*,\ldots,\sigma_n^*)$, which we denote by
\[ ((\sigma_1,\ldots,\sigma_n),t)=\Compress((\sigma_1^*,\ldots,\sigma_n^*)). \]

The following result follows directly from the definitions, and
implies that $(\Compress\circ\Expand)$ and $(\Compress\circ\Expand)$
are identities.

\begin{lemma}
  Let $(\sigma_1,\ldots,\sigma_n)$ be a tuple of traces of $\krip$,
  $t$ a trajectory and $(\sigma_1^*,\ldots,\sigma_n^*)$ a tuple of
  fair traces of $\kripST$. Then
  $((\sigma_1,\ldots,\sigma_n),t)=\Compress((\sigma_1^*,\ldots,\sigma_n^*))$
  if and only if
  $(\sigma_1^*,\ldots,\sigma_n^*)=\Expand((\sigma_1,\ldots,\sigma_n),t)$.
\end{lemma}

Abusing notation, we use $\Pi'=\Expand(\Pi,t)$ for trace assignments
$\Pi$ in $\krip$, trajectory $t$ an $\Pi'$ in $\kripST$ if
$(\Pi'(\pi_1),\ldots,\Pi'(\pi_n))=\Expand((\Pi(\pi_1),\ldots,\Pi(\pi_n)),t)$.
We define $(\Pi,t)=\Compress(\Pi')$ analogously.

We first prove some auxiliary lemmas.

\begin{lemma}
  \label{lem:four}
  $(\Pi,t)\models\psiPH$ if and only if
  $\Expand(\Pi,t)\models(\Always\Phases \And\psiPH)$.
\end{lemma}

\begin{proof}
  The proof (in both directions) proceeds by a simple induction on the
  steps $k$ of the traces. \qed
\end{proof}


\begin{lemma}
  \label{lem:five}
  Let $\Pi'$ and $\Pi''$ be stuttering equivalent with
  $\Pi'\models\Always\Phases$ and $\Pi'\models\Always\Phases$.
  Then $\Pi'\models\psiPH$ if and only if $\Pi''\models\psiPH$.
\end{lemma}

\begin{proof}
  The proof proceeds by proving that the valuations of $\psiPH$ at
  color-changing positions and non color-changing positions coincide
  for all atomic phase formulas. \qed
\end{proof}

\begin{lemma}
  \label{lem:six}
  Let $\Pi'$ be such that
  $\Pi'\models\Phases\,\U(\Bad\Or\Block)$ and let $\Pi''$ be
  such that $\Pi'$ and $\Pi''$ are stuttering equivalent.
  Then $\Pi''\not\models\Always\Phases$.
\end{lemma}

\begin{proof}
  Since $\Pi'\models\Phases\,\U(\Bad\Or\Block)$ let $k$ be the
  position at which $(\Pi'+k)\models\Bad\Or\Block$.
  We consider the two cases separately.
  \begin{itemize}
  \item $(\Pi'+k)\models\Bad$. Let $\pi_i$ and $\pi_j$ be the paths
    that are miss-aligned.
    This means that in from $\Pi'$ and $\Pi'+k$, $\Pi'(\pi_i)$ and
    $\Pi'(\pi_j)$ change $P$-phases the same number of times, and one
    of them (say $\pi_i$) will not change color and the other will.
    If $\Pi''$ is stuttering equivalent to $\Pi'$, then $\Pi''(\pi_i)$
    and $\Pi''(\pi_j)$ will also change a different colors a different
    number of times, which contradicts that
    $\Pi''\models\Always\Phases$.
  \item $(\Pi'+k)\models\Block$. In this case, the cycle implies each
    trace involved performs at $k$, a fix-number of color-changes at
    its $k$-th position for each of its each two atomic phase
    formulas, and that the circular constraint is mutually
    unsatisfiable.
    IF $\Pi''$ is stuttering equivalent to $\Pi'$, $\Always\Phases$
    would imply that each trace involved in the cycle will be with the
    same number of color changes after reaching the position
    equivalent to $k$ in $\Pi'$, therefore violating the constraint as well.\qed
  \end{itemize}    
\end{proof}

A consequence of Lemma~\ref{lem:four} and Lemma~\ref{lem:six} is the
following corollary.

\begin{corollary}
  \label{cor:seven}
  $\Pi'\models\neg(\Phases\,\U(\Bad\Or\Block))$ and
  $(\Pi,t)=\Compress(\Pi')$ then there is no trajectory $t_2$ for which
  $(\Pi,t_2)\models\psiPH$.
\end{corollary}

\begin{proof} Otherwise, Lemma~\ref{lem:four} would imply that
  $\Expand(\Pi,t_2)\models\Always\Phases\And\psiPH$ and
  $\Expand(\Pi,t_2)$ is a stuttering equivalent to $\Pi'$.\qed
\end{proof}  

We are ready to prove the main result.

\setcounter{aux}{\value{theorem}}
\setcounter{theorem}{\value{thm-sync-async}}

\begin{theorem}
  Let $\krip$ be a Kripke structure and $\psi$ an admissible formula.
  Then $\krip\models \forall\pi_1\ldots\pi_n.\Etau.\psi$ if and only if
  $\kripST\models \forall\pi_1\ldots\pi_n.\psiSync$.
\end{theorem}

\setcounter{theorem}{\value{aux}}

\begin{proof}
  We prove the two directions separately.
  \begin{itemize}
  \item ``$\Rightarrow$''. Assume
    $\krip\models \forall\pi_1\ldots\pi_n.\Etau.\psi$.
    Let $\Pi'$ be an arbitrary trace assignment for $\kripST$.
    By contradiction, assume $\Pi'\not\models\psiSync$.
    Let $(\Pi,t)=\Compress(\Pi')$.
    By our assumption, there must exists a $t_2$ such that
    $(\Pi,t_2)\models\psi$.
    Since the valuation for the state and monadic sub-formulas of
    $\psi$ is the same for $(\Pi,t)$, $(\Pi,t_2)$ and
    $\Pi'$, it must be the case that
    $\Pi'\not\models\psi'$ and $(\Pi,t_2)\models\psiPH$.
    Now, if $\Pi'\not\models\psi'$, then either
    \begin{itemize}
    \item $\Pi'\models\Phases\,\U(\Bad\Or\Block)$. In this case
      Corollary~\ref{cor:seven} contradicts that
      $(\Pi,t_2)\models\psiPH$.
    \item $\Pi'\not\models\Always\Phases\Then\psiPH$. In other words
      $\Pi'\models\Always\Phases$ but
      $\Pi'\not\models\psiPH$.
      However, Lemma~\ref{lem:five} implies that $\Pi'\models\psiPH$
      because $\Expand(\Pi,t_2)$ and $\Pi'$ are stuttering equivalent and
      $\Expand(\Pi,t_2)\models\Always\Phases\And\psiPH$.
    \end{itemize}
  \item ``$\Leftarrow$''. We assume now that
    $\kripST\models \forall\pi_1\ldots\pi_n.\psiSync$.
    Let $\Pi$ be an arbitrary trace assignment of $\krip$ and let 
    let $\Pi'$ be an arbitrary fair stuttering expansion.
    Then, by assumption $\Pi'\models\psiSync$.
    \begin{itemize}
    \item If $\Pi'\not\models\psi'$ then (all state and monadic sub-formulas
      of $\psi$ coincide for $\Pi$ and $\Pi'$), then
      $(\Pi,t)\models\psi$ for any $t$ and we are done.
    \item If $\Pi'\models\psi'$, in particular
      $\Pi'\models\Always\Phases\And\psiPH$. Then, by
      Lemma~\ref{lem:four}, $\Compress(\Pi')\models\psiPH$ and
      $\Compress(\Pi')=(\Pi,t)$ which witnesses the appropriate trajectory $t$.\qed
    \end{itemize}
  \end{itemize}
\end{proof}

\setcounter{aux}{\value{theorem}}
\setcounter{theorem}{\value{thm-EE}}

\begin{theorem}
  Let $\krip$ be a Kripke structure and $\psi$ an admissible formula.
  Then $\krip\models \exists\pi_1\ldots\pi_n.\Etau.\psi$ if and only if
  $\kripST\models \exists\pi_1\ldots\pi_n.\psiESync$.
\end{theorem}

\begin{proof}
  This follows directly from Lemma~\ref{lem:four}.\qed
\end{proof}

\setcounter{aux}{\value{theorem}}
\setcounter{theorem}{\value{thm-adm-coadm}}

\begin{theorem}
  Model-checking $\forall^*$ or $\exists^*$ admissible and
  co-admissible formulas is decidable both for formulas with $\Etau$
  and formulas with $\Atau$.
\end{theorem}

\setcounter{theorem}{\value{aux}}





